\documentclass[12pt]{article}

\usepackage{geometry}
 \usepackage{amsthm}

 \newtheoremstyle{theoremdd}
  {0}
  {0}
  {\itshape}
  {0pt}
  {\bfseries}
  {}
  { }
  {\thmname{#1}\thmnumber{ #2}\textnormal{\thmnote{ (#3)}}}

\theoremstyle{theoremdd}

\usepackage{setspace,multirow}

\makeatletter
\newcommand*{\rom}[1]{\expandafter\@slowromancap\romannumeral #1@}
\makeatother

\usepackage{lmodern,enumerate,rotating,anyfontsize}
\usepackage{amsmath,amsfonts,amssymb,subcaption,mathtools,blkarray,soul,etoolbox}
\usepackage{tikz-cd}
\usetikzlibrary{cd}
\usetikzlibrary{shapes,shapes.geometric,arrows,fit,calc,positioning,automata}

\usepackage{stmaryrd}
\SetSymbolFont{stmry}{bold}{U}{stmry}{m}{n}

\DeclareMathOperator{\CCa}{CC_A}

\allowdisplaybreaks

\DeclareMathOperator{\obs}{obs}

\DeclareMathOperator{\dss}{dss}

\newcommand{\Scal}{\mathcal{S}}

\newtheorem{theorem}{Theorem}[section]
\newtheorem{lemma}[theorem]{Lemma}

\newtheorem{definition}{Definition}
\newtheorem{proposition}[theorem]{Proposition}

\newtheorem{example}[theorem]{Example}

\newtheorem{remark}{Remark}

\newtheorem{problem}{Problem}
\usepackage{times,amssymb,amsfonts,mathtools,graphicx,tikz,algorithm,algorithmic,url,hhline,booktabs}
\usetikzlibrary{automata,arrows,positioning}

\usetikzlibrary{shadows}

\makeatletter
\let\NAT@parse\undefined
\makeatother
\usepackage[colorlinks,linkcolor=blue,anchorcolor=blue,citecolor=green,urlcolor=cyan,hyperindex]{hyperref}

\usetikzlibrary{automata,arrows,positioning}
\usetikzlibrary{petri}

\newcounter{enumi_saved}

\tikzset{elliptic state/.style={draw,ellipse}}
\tikzset{rectangle state/.style={draw,rectangle}}

\tikzset{
node distance=3cm, 
every state/.style={thick, fill=gray!10}, 
initial text=$ $, 
}

\newcommand{\red}{\color{red}}

\definecolor{green}{rgb}{0.1,0.7,0.1}

\newcommand{\algorithmicstop}{\textbf{stop}}
\newcommand{\STOP}{\STATE \algorithmicstop}

\usepackage{tcolorbox}
\usepackage{tabularx}
\usepackage{array}
\usepackage{colortbl}
\tcbuselibrary{skins}

\title{Verification and Enforcement of Strong State-Based Opacity for Discrete-Event Systems}

\author{Xiaoguan Han\\
{\small College of Electronic Information and Automation}\\
{\small Tianjin University of Science and Technology, Tianjin 300222, China}\\
{\small hxg-allen@163.com}
\and
Kuize Zhang\\
{\small Department of Computer Science}\\
{\small University of Surrey, Guildford GU2 7XH, United Kingdom}\\
{\small kuize.zhang@surrey.ac.uk}
\and 
Zhiwu Li\\
{\small Institute of Systems Engineering}\\
{\small Macau University of Science and Technology, Taipa 999078, China}\\
{\small School of Electro-Mechanical Engineering}\\
{\small Xidian University, Xi'an 710071, China}\\
{\small zhwli@xidian.edu.cn}
}

\begin{document}

\date{}

\maketitle

{\bf Abstract}
  In this paper, we investigate the verification and enforcement of strong state-based opacity (SBO) in discrete-event systems modeled as partially-observed (nondeterministic) finite-state automata, including strong $K$-step opacity ($K$-SSO), strong current-state opacity (SCSO), strong initial-state opacity (SISO), and strong infinite-step opacity (Inf-SSO).
  They are stronger versions of four widely-studied standard opacity notions, respectively.
  We firstly propose a new notion of $K$-SSO, and then we construct a concurrent-composition structure that is a variant of our previously-proposed one to verify it.
  Based on this structure, a verification algorithm for the proposed notion of $K$-SSO is designed.
  Also, an upper bound on $K$ in the proposed $K$-SSO is derived.
  Secondly, we propose a distinctive opacity-enforcement mechanism that has better scalability than the existing ones (such as supervisory control).
  The basic philosophy of this new mechanism is choosing a subset of controllable transitions to disable before an original system starts to run in order to cut off all its runs that violate a notion of strong SBO of interest.
  Accordingly, the algorithms for enforcing the above-mentioned four notions of strong SBO are designed using the proposed two concurrent-composition structures.
  In particular, the designed algorithm for enforcing Inf-SSO has lower time complexity than the existing one in the literature, and does not depend on any assumption.
  Finally, we illustrate the applications of the designed algorithms using examples.


\tableofcontents

\section{Introduction}\label{sec1}

Opacity is a concealment property, which guarantees that the secret information of a system cannot be distinguished from its non-secret information to an outside observer called an \emph{intruder} who knows complete knowledge of the system but can only observe its outputs.
In other words, an opaque system always holds the ``plausible deniability" for its secrets during its execution.
Opacity adapts to the characteristics of a variety of security and privacy requirements in diverse dynamic systems, including event-driven systems~\cite{Lafortune(2018),Hadjicostis(2020)}, time-driven systems~\cite{An(2020),Ramasubramanian(2020)}, and metric systems~\cite{Yin(2021)}.

The notion of opacity initially appeared in the computer science literature~\cite{Mazare(2004)} for analyzing cryptographic protocols.
Whereafter, various notions of opacity were introduced in the context of discrete-event systems (DESs), including Petri nets~\cite{Bryans(2005)}, automata~\cite{Badouel(2007),Saboori(2007)}, labelled transition systems~\cite{Bryans(2008)}, etc.
For details see the recent surveys~\cite{Lafortune(2018),Jacob(2016)} and the textbook~\cite{Hadjicostis(2020)}.
In plain words, a system is called \emph{opaque} if for each behavior relevant to its secrets, there always exists an observationally-equivalent non-secret behavior.
Note that, for a given DES, an intruder completely knows the system's structure but has only limited observations to its behavior.
In general, the secrets of a DES are modeled by two ways: 1) a set of secret states, and 2) a set of secret traces.
For the former, opacity is referred to as state-based (e.g.,~\cite{Bryans(2005),Saboori(2007)}), while for the latter, opacity is referred to as language-based (e.g.,~\cite{Badouel(2007),Bryans(2008)}).

In automata-based formalisms, different notions of opacity were proposed, including current-state opacity (CSO)~\cite{Saboori(2007)}, initial-state opacity (ISO)~\cite{Saboori(2013)}, $K$-step opacity ($K$-SO)~\cite{Saboori(2011a)}, infinite-step opacity (Inf-SO)~\cite{Saboori(2012)}\footnote{For convenience, the notion originally called CSO (resp., ISO, $K$-SO, and Inf-SO) is renamed as standard CSO (resp., standard ISO, standard $K$-SO, and standard Inf-SO) in this paper.}, and language-based opacity (LBO)~\cite{Badouel(2007),Bryans(2008),Lin(2011)}.
Some efficient algorithms to check them have also been provided in~\cite{Wu(2013),Zhang(2023a),Balun(2023)}.
In particular, it was proven that the above-mentioned five versions of opacity can be reduced to each other in polynomial time when LBO is restricted to the special case that the secret languages are regular (cf.,~\cite{Wu(2013),Balun(2021)}), while LBO is generally undecidable in finite-state automata with $\epsilon$-labeling functions (cf.,~\cite{Bryans(2008)}).
Furthermore, when a system is not opaque, a natural question is to ask ``how can one enforce the system to be opaque?
This is the so-called \emph{opacity enforcement}, which has been widely investigated using a variety of techniques, including supervisory control~\cite{Dubreil(2010),Saboori(2012b),Yin(2016),Tong(2018),Moulton(2022)}, insertion/edit functions~\cite{Wu(2014a),Ji(2018),Ji(2019),Liu(2022),Li(2022)}, dynamic observers~\cite{Cassez(2012),Zhang(2015),Yin(2020)}, etc.
In addition, verification and/or enforcement of opacity have been extended to other classes of models, see, e.g.,~\cite{Tong(2017),Zhang(2019),Keroglou(2017),Yin(2019a),Deng(2020),Zhang(2024),Hou(2022),Zhang(2023)}.
Some applications of opacity in real-world systems have also been provided in the literature, see, e.g.,~\cite{Saboori(2011b),Wu(2014),Bourouis(2017),Lin(2020)}.

Among various notions of opacity, the standard CSO guarantees that an intruder cannot make sure whether a system is currently at a secret state.
In Location-Based Services (LBS), a user may want to hide his/her initial location or his/her location (e.g., visiting a hospital or bank) at some specific previous instants.
Such requirements can be characterized by the standard ISO and $K$/Inf-SO.
However, as mentioned in~\cite{Falcone(2015),Ma(2021),Han(2023)}, these four notions of standard opacity have some limitations in practice.
Specifically, they cannot capture such a scenario that an intruder can never infer for sure whether a system has passed through a secret state based on his/her observations.
In other words, even though a system is ``\emph{opaque}" in the standard sense, the intruder may still be able to determine that a secret state has been passed through.
This motivated the authors in~\cite{Falcone(2015)} to introduce the notion of \emph{strong $K$-step opacity} ($K$-SSO), which is a stronger version of the standard $K$-SO and captures that the visit of a secret state cannot be inferred within the last $K$ observable steps.
Inspired by~\cite{Falcone(2015)}, the notion of $K$-SSO was extended to \emph{strong infinite-step opacity} (Inf-SSO) in~\cite{Ma(2021)}, which is a stronger version of the standard Inf-SO.
Further, two algorithms have been presented to verify $K$-SSO and Inf-SSO using the so-called $K$-step and $\infty$-step recognizers, respectively.
In particular, the algorithm for verifying $K$-SSO is more efficient than that of~\cite{Falcone(2015)}.
Later on, the verification efficiency of $K$-SSO proposed in~\cite{Falcone(2015)} was further improved using our proposed concurrent-composition approach in~\cite{Zhang(2023a)}.

Recently, in our previous work~\cite{Han(2023)}, two stronger versions of the standard CSO and ISO, called \emph{strong current-state opacity} (SCSO) and \emph{strong initial-state opacity} (SISO), were proposed in nondeterministic finite-state automata, respectively.
SCSO (resp., SISO) guarantees that for each run generated by a system ending (resp., starting) at a secret state, there always exists an observationally-equivalent non-secret run.
Further, we developed a novel methodology to simultaneously verify SCSO, SISO, and Inf-SSO using a concurrent-composition technique.
It was also proven that our designed algorithm for verifying Inf-SSO has lower time complexity than that in~\cite{Falcone(2015)} and~\cite{Ma(2021)}.
Additionally, in~\cite{Zhang(2023b)} we proposed another type of $K$-SSO notion, which is slightly stronger than that in studied in~\cite{Falcone(2015)} and~\cite{Ma(2021)}.
Note that, the notion of $K$-SSO proposed in~\cite{Zhang(2023b)} can be directly checked using our concurrent-composition structure in~\cite{Han(2023)}, see~\cite{Zhang(2023a),Zhang(2023b)} for details.

To avoid confusion, from now on we classify the various notions of opacity in the context of finite-state automata into two categories:
1) \textbf{language-based opacity} (LBO), including strong LBO and weak LBO (cf.,~\cite{Lin(2011)});
and 2) \textbf{state-based opacity} (SBO), including standard SBO (CSO, ISO, $K$-SO, Inf-SO) and strong SBO (SCSO, SISO, $K$-SSO, Inf-SSO).

As an extension of~\cite{Han(2023)} and~\cite{Zhang(2023b)}, in this paper we propose another new notion of $K$-SSO, which is weaker than both two notions of $K$-SSO proposed in~\cite{Falcone(2015)} and~\cite{Zhang(2023b)} but stronger than standard $K$-SO.
Further, a new concurrent-composition structure that is a variant of that of~\cite{Han(2023)} is constructed to do efficient verification for the proposed new notion of $K$-SSO.
And then, we develop a distinctive-enforcement mechanism to enforce various notions of strong SBO.
Specifically, the main contributions of this paper are as follows.

\begin{itemize}
\item Although the concurrent-composition structure provided in our previous work~\cite{Han(2023)} can efficiently verify SCSO, SISO, Inf-SSO, and $K$-SSO defined in~\cite{Zhang(2023b)}, it cannot be
  directly used to do verification for the proposed new notion of $K$-SSO.
  For this purpose, we propose a variant of the concurrent-composition structure given in~\cite{Han(2023)}.
  The proposed variant structure does not depend on $K$.
  Further, we exploit it to design an efficient verification algorithm for the proposed notion of $K$-SSO, which has time complexity $\mathcal{O}((|\Sigma_o||\Sigma_{uo}|+|\Sigma|)|X|^{2}2^{|X|})$.
  Moreover, an upper bound on $K$ in the proposed $K$-SSO is derived, i.e., $|\hat{X}|2^{|X\backslash X_S|}-1$, where $|\hat{X}|$ is the number of states of
  the initial-secret subautomaton $\hat{G}$.
  This results in the fact that a system $G$ is $K$-SSO proposed in this paper if and only if it is min$\{K, |\hat{X}|2^{|X\backslash X_S|}-1\}$-SSO.

\item We propose a distinctive opacity-enforcement mechanism that also works in a closed-loop manner and restricts a system's behavior to ensure that the modified system does not reveal its ``secrets" to an intruder.
  Specifically, this new opacity-enforcement approach chooses a subset of controllable transitions to disable before an original system starts to run in order to cut off all runs that violate a notion of strong SBO of interest.
  This remarkably differs from the existing opacity-enforcement methods.

\item Under the proposed control mechanism, we consider the enforcement of $K$-SSO, SCSO, SISO, and Inf-SSO, respectively.
  We use the proposed concurrent-composition structures to design an opacity-enforcement algorithm for each of these four strong SBO notions.
  Our proposed concurrent-composition approach does not depend on any assumption.
  By contrast, the containment relationship between $\Sigma_o$ and $\Sigma_c$ is widely used in the supervisory control framework.
  The proposed algorithms of enforcing these four strong SBO notions have the same (worst-case) time complexity as their verification, and are currently the most efficient.
  This means that the proposed algorithm for enforcing Inf-SSO leads to a considerable improvement compared with the existing one in~\cite{Ma(2021)} under the supervisory control framework from the complexity's point of view.
\end{itemize}

\begin{remark}\label{re:1.1}
Let us give some comments on the differences between the proposed opacity-enforcement mechanism and the existing ones in the DES literature.
Until now, there are mainly four different types of opacity-enforcement mechanisms: supervisory control~\cite{Dubreil(2010),Saboori(2012b),Yin(2016),Tong(2018),Moulton(2022)}, insertion/edit functions~\cite{Wu(2014a),Ji(2018),Ji(2019),Liu(2022),Li(2022)}, dynamic observers~\cite{Cassez(2012),Zhang(2015),Yin(2020)}, and runtime enforcement approach~\cite{Falcone(2015)}.
Among these four types, supervisor control is the most relevant to our approach and probably also the most widely used for the enforcement of opacity.
However, there are intrinsic differences between them.
Specifically, the classical supervisor control framework chooses controllable events to disable according to the observed output sequences.
Whereas the proposed enforcement approach is to choose controllable transitions to disable before a system starts to run.
Moreover, our approach can choose a subset of controllable transitions to disable and can also add additional transitions, so it is more flexible than supervisory control.

Note that, this new enforcement mechanism initially appeared in our previous work~\cite{Zhang(2023c)} for enforcing delayed strong detectability of DESs.
However, enforcement of opacity has not been considered in~\cite{Zhang(2023c)}.
To the best of our knowledge, this paper is the first to systematically investigate opacity enforcement under this new enforcement framework, which will be discussed in more detail in Section~\ref{sec4} below.
\end{remark}

The rest of this paper is arranged as follows.
Section~\ref{sec2} introduces the system model adopted in this paper, relevant notations, and basic notions.
In Section~\ref{sec3}, we propose a new notion of $K$-SSO.
To verify it, a new concurrent-composition structure is proposed, based on which we present an efficient verification algorithm, as well as an upper bound on the value of $K$ in the proposed notion of $K$-SSO.
In Section~\ref{sec4}, we propose a novel opacity-enforcement approach for $K$-SSO, SCSO, SISO, and Inf-SSO.
Accordingly, the enforcement algorithms for all notions of strong SBO are designed and their time complexity is also analyzed, respectively.
Finally, we conclude this paper in Section~\ref{sec5}.

\section{Preliminaries}\label{sec2}

A \emph{nondeterministic finite-state automaton} (NFA) is a quadruple $G=(X,\Sigma,\delta, X_0)$, where $X$ is a finite set of states, $\Sigma$ is a finite alphabet of events, $X_0\subseteq X$ is a set of initial states, $\delta: X\times\Sigma\rightarrow 2^{X}$ is the transition function, which depicts the dynamics of $G$: given states $x,y\in X$ and an event $\sigma\in\Sigma$, $y\in\delta(x,\sigma)$ implies that there exists a transition labeled by $\sigma$ from $x$ to $y$, i.e., $x\stackrel{\delta}{\rightarrow}y$.
We can extend the transition function to $\delta :X\times\Sigma^{\ast}\rightarrow 2^{X}$ in the recursive manner, where $\Sigma^{\ast}$ denotes the \emph{Kleene closure} of $\Sigma$, consisting of all finite sequences composed of the events in $\Sigma$ (including the empty sequence $\epsilon$)\footnote{Note that, $\delta(x,\epsilon):=\{x\}$ for any $x\in X$.}. Details can be found in~\cite{Hopcroft(2001)}.
We use $\mathcal{L}(G,x)$ to denote the language generated by $G$ from state $x$, i.e., $\mathcal{L}(G,x)=\{s\in\Sigma^{\ast}: \delta(x,s)\neq\emptyset\}$.
Therefore, the language generated by $G$ is $\mathcal{L}(G)=\cup_{x_0\in X_0}\mathcal{L}(G,x_0)$.
$G$ is called a deterministic finite-state automation (DFA) if $|X_0|=1$ and $|\delta(x,\sigma)|\leq 1$ for all $x\in X$ and all $\sigma\in\Sigma$.
When $G$ is deterministic, $\delta$ is also regarded as a partial transition function $\delta: X\times\Sigma^{\ast}\rightarrow X$.
For a sequence $s\in\mathcal{L}(G)$, we denote its length by $|s|$ and its prefix closure by $Pr(s)$, i.e., $Pr(s)=\{w\in\mathcal{L}(G): (\exists w^\prime\in\Sigma^{\ast})[ww^\prime=s]\}$.
Further, for a prefix $w\in Pr(s)$, we use the notation $s/w$ to denote the suffix of $s$ after its prefix $w$.

In this paper, a DES of interest is modeled as an NFA $G$.
As usual, we assume that an intruder can only see partial behaviors of $G$.
To this end, $\Sigma$ is partitioned into the set $\Sigma_o$ of observable events and the set $\Sigma_{uo}$ of unobservable events, i.e., $\Sigma_o\cup\Sigma_{uo}=\Sigma$ and $\Sigma_o\cap\Sigma_{uo}=\emptyset$.
The natural projection $P:\Sigma^{\ast}\rightarrow \Sigma_o^{\ast}$ is defined recursively by:
($i$) $P(\epsilon)=\epsilon$,
($ii$) $P(s\sigma)=P(s)\sigma,\mbox{ if } \sigma\in \Sigma_o$,
and ($iii$) $P(s\sigma)=P(s),\mbox{ if } \sigma\in \Sigma_{uo}$, where $s\in\Sigma^\ast$.
We extend the natural projection $P$ to $\mathcal{L}(G)$ by $P(\mathcal{L}(G))=\{P(s)\in\Sigma_{o}^{\ast}: s\in\mathcal{L}(G)\}$, see~\cite{Cassandras(2021)} for details.
Without loss of generality, we assume that system $G$ is accessible, i.e., all its states are reachable from $X_0$.
A state $x\in X$ is called \emph{$K$-step observationally reachable} if there exists an initial state $x_0\in X_0$ and a sequence $s\in\mathcal{L}(G,x_0)$ such that
$x\in\delta(x_0,s)$ and $|P(s)|=K$, where $K\in\mathbb{N}$ is a non-negative integer.

To study strong SBO of $G=(X,\Sigma,\delta,X_0)$, it is assumed that $G$ has a set of secret states, denoted by $X_{S}\subseteq X$.
We denote its the set of non-secret states by $X_{NS}=X\setminus X_S$.
We use the notation $X^{S}_0$ to denote the set of secret initial states (i.e., $X^{S}_0=X_0\cap X_S$).
The set of non-secret initial states is denoted by $X^{NS}_0$, i.e., $X^{NS}_0=X_0\setminus X^{S}_0$.
Consider an $n$-length sequence $s=\sigma_{1}\sigma_{2}\cdots\sigma_{n}\in\Sigma^{\ast}$ (where $n\in\mathbb{N}$), $x_0\in X_0$, and $x_i\in X$, $i=1,2,\ldots,n$, if $x_{k+1}\in\delta(x_k,\sigma_{k+1})$, $0\leq k\leq n-1$,
we call $x_0\stackrel{\sigma_1}{\rightarrow}x_1\stackrel{\sigma_2}{\rightarrow}x_2\stackrel{\sigma_3}{\rightarrow}\cdots\stackrel{\sigma_n}{\rightarrow}x_n$ a \emph{run} generated by $G$ from $x_0$ to $x_n$ under $s$.
For brevity, we write $x_0\stackrel{s}{\rightarrow}x_n$ (resp., $x_0\stackrel{s}{\rightarrow}$) when $x_1,x_2,\ldots,x_{n-1}$ (resp., $x_1,x_2,\ldots,x_n$) are not specified.
Note that $x_0\stackrel{s}{\rightarrow}x_n$ (resp., $x_0\stackrel{s}{\rightarrow}$) may denote more than one run based on the nondeterminism of $G$, which depends on the context.
A run $x_0\stackrel{\sigma_1}{\rightarrow}x_1\stackrel{\sigma_2}{\rightarrow}x_2\stackrel{\sigma_3}{\rightarrow}\cdots\stackrel{\sigma_n}{\rightarrow}x_n$ (resp., $x_0\stackrel{\sigma_1}{\rightarrow}x_1\stackrel{\sigma_2}{\rightarrow}x_2\stackrel{\sigma_3}{\rightarrow}\cdots$), abbreviated as $x_0\stackrel{s}{\rightarrow}x_n$ (resp., $x_0\stackrel{s}{\rightarrow}$), is called \emph{non-secret} if all $x_i\in X_{NS}$, $i=0,1,2,\ldots,n$ (resp., $i=0,1,2,\ldots$); otherwise, it is secret.

\section{Verification of strong \texorpdfstring{$K$}{K}-step opacity}\label{sec3}

\subsection{The notion of strong\texorpdfstring{$K$}{K}-step opacity }\label{subsec3.1}

In~\cite{Falcone(2015)}, the authors proposed a notion of strong $K$-step opacity ($K$-SSO) for a DFA $G$ to overcome the limitation that the notion of standard $K$-SO is not sufficiently strong (i.e., the standard $K$-SO cannot guarantee that the visit of a secret state cannot be inferred within the last $K$ observable steps).
Now we generalize the definition of $K$-SSO proposed in~\cite{Falcone(2015)} in nondeterministic finite-state automata (NFAs).

\begin{definition}[\cite{Falcone(2015)}]\label{de:3.1a}
Given a system $G=(X,\Sigma,\delta,X_0)$, a projection map $P$ w.r.t. a set $\Sigma_o$ of observable events, a set $X_{S}\subseteq X$ of secret states, and a non-negative integer $K$,
$G$ is said to be strongly $K$-step opaque ($K$-SSO)\footnote{In this paper, the terminology ``$K$-SSO" is the acronym of both ``strong $K$-step opacity" and ``strongly $K$-step opaque", which depends on the context.} w.r.t. $\Sigma_o$ and $X_{S}$ if for all $x_0\in X_0$ and all $s\in\mathcal L(G,x_0)$, there exists a $x^\prime_0\in X_0$ and a $w\in\mathcal L(G,x^\prime_0)$ with $P(w)=P(s)$ such that for all $\bar{w}\in Pr(w)$ satisfying $|P(w)|-|P(\bar{w})|\leq K$, $\delta(x^\prime_0,\bar{w})\subseteq X_{NS}$.
\end{definition}

In plain words, $K$-SSO in Definition~\ref{de:3.1a} implies that in the last at most $K$ observable steps, an intruder cannot make sure whether some secret states have been visited.
Note that, the notion of $K$-SSO studied in~\cite{Ma(2021)} is equivalent to Definition~\ref{de:3.1a} when system $G$ is deterministic.
Currently, the most efficient verification algorithm for $K$-SSO in Definition~\ref{de:3.1a} has been designed in our previous work~\cite{Zhang(2023a)}.

Recently, in~\cite{Zhang(2023b)} we proposed another type of $K$-SSO notion (see Definition 2.25 therein), which can be regarded as strengthening of SCSO notion (see Definition~\ref{de:4.2} below).
Now we recall the notion of $K$-SSO formulated by Definition 2.25 in~\cite{Zhang(2023b)}.

\begin{definition}[\cite{Zhang(2023b)}]\label{de:3.1b}
Given a system $G=(X,\Sigma,\delta,X_0)$, a projection map $P$ w.r.t. a set $\Sigma_o$ of observable events, a set $X_{S}\subseteq X$ of secret states, and a non-negative integer $K$,
$G$ is said to be strongly $K$-step opaque ($K$-SSO) w.r.t. $\Sigma_o$ and $X_{S}$ if
\begin{align}\label{eq:3.1b}
&(\forall\mbox{ run } x_0\stackrel{s_1}{\rightarrow}x_1\stackrel{s_2}{\rightarrow}x_2: x_0\in X_0\wedge x_1\in X_S\wedge|P(s_2)|\leq K)\nonumber\\
&(\exists\mbox{ non-secret run }x^\prime_0\stackrel{s^\prime_1}{\rightarrow}{x^\prime_1}\stackrel{s^\prime_2}{\rightarrow}{x^\prime_2})[(x^\prime_0\in X_0)\wedge (P(s^\prime_1)=\nonumber\\
&P(s_1))\wedge (P(s^\prime_2)=P(s_2))].
\end{align}
\end{definition}

By Definitions~\ref{de:3.1b} and~\ref{de:4.2}, it is not difficult to see that $K$-SSO in Definition~\ref{de:3.1b} reduces to SCSO in Definition~\ref{de:4.2} when string $t=\epsilon$.
Further, $K$-SSO in Definition~\ref{de:3.1b} can be directly checked using our concurrent-composition structure presented in~\cite{Han(2023)}, see the Theorem $2.15$ of~\cite{Zhang(2023b)} for details.

Next, in this subsection we propose another new notion of $K$-SSO, which is weaker than both two notions of $K$-SSO proposed in~\cite{Falcone(2015)} and~\cite{Zhang(2023b)} but stronger than standard $K$-SO.

\begin{definition}\label{de:3.1}
Given a system $G=(X,\Sigma,\delta,X_0)$, a projection map $P$ w.r.t. a set $\Sigma_o$ of observable events, a set $X_{S}\subseteq X$ of secret states, and a non-negative integer $K$,
$G$ is said to be strongly $K$-step opaque ($K$-SSO) w.r.t. $\Sigma_o$ and $X_{S}$ if
\begin{align}\label{eq:3.1}
&(\forall\mbox{ run } x_0\stackrel{s_1}{\rightarrow}x_1\stackrel{s_2}{\rightarrow}x_2: x_0\in X_0\wedge x_1\in X_S\wedge|P(s_2)|\leq K)\nonumber\\
&(\exists\mbox{ run } x_0^\prime\stackrel{s_1^\prime}{\rightarrow}{x_1^\prime}\stackrel{s_2^\prime}{\rightarrow}{x_2^\prime})[(x_0^\prime\in X_0)\wedge(P(s_1^\prime)=P(s_1))\wedge\nonumber\\
&(P(s_2^\prime)=P(s_2))\wedge({x_1^\prime}\stackrel{s_2^\prime}{\rightarrow}{x_2^\prime}\mbox{ is non-secret})].
\end{align}
\end{definition}

\begin{remark}\label{re:3.1}
1) For $K$-SSO in Definition~\ref{de:3.1}, if no such a run $x_0^\prime\stackrel{s_1^\prime}{\rightarrow}{x_1^\prime}\stackrel{s_2^\prime}{\rightarrow}{x_2^\prime}$ exists,
we call $x_0\stackrel{s_1}{\rightarrow}x_1\stackrel{s_2}{\rightarrow}x_2$ a \emph{leaking-secret run}.
Similarly, we also define a leaking-secret run for other notions of opacity.

2) If a system $G$ is $K$-SSO in Definition~\ref{de:3.1}, then an intruder cannot be sure whether the system is/was in a secret state within the last $K$ observable steps.
We compare these aforementioned three notions of $K$-SSO.
From their definitions, we conclude readily that: 1) $K$-SSO in Definition~\ref{de:3.1b} implies $K$-SSO in Definition~\ref{de:3.1a}, but the converse is not true; 2) $K$-SSO in Definition~\ref{de:3.1a} implies $K$-SSO in Definition~\ref{de:3.1}, but the converse is also not true.
However, these three notions of $K$-SSO are obviously stronger than the standard $K$-SO proposed in~\cite{Saboori(2011a)}.
\end{remark}

\subsection{The structure of concurrent composition}\label{subsec3.2}

In this subsection, we propose a new information structure using a concurrent-composition approach to verify $K$-SSO in Definition~\ref{de:3.1}.
Note that, the proposed information structure is a variant of that in our previous work~\cite{Han(2023)}.
In order to present this structure, we need to introduce two notions of \emph{initial-secret subautomaton} and \emph{non-secret subautomaton}.

Given a system $G=(X,\Sigma,\delta,X_0)$ and a set $X_S\subseteq X$ of secret states.
We first recall the notion of standard \emph{subset construction} of $G$ also called an \emph{observer}, which is defined by
\begin{equation}\label{eq:3.2}
Obs(G)=(X_{obs},\Sigma_{obs},\delta_{obs},X_{obs,0}),
\end{equation}
where $X_{obs}\subseteq 2^X\backslash\{\emptyset\}$ stands for the set of states,
$\Sigma_{obs}=\Sigma_o$ stands for the set of observable events,
$\delta_{obs}: X_{obs}\times\Sigma_{obs}\rightarrow X_{obs}$ stands for the (partial) deterministic transition function defined as follows: for any $q\in X_{obs}$ and $\sigma\in\Sigma_{obs}$,
we have $\delta_{obs}(q,\sigma)=\{x^{\prime}\in X:\exists x\in q, \exists w\in\Sigma_{uo}^\ast\mbox{ s.t. }x^{\prime}\in\delta(x,\sigma w)\}$ if it is nonempty,
$X_{obs,0}=\{x\in X:\exists x_0\in X_0, \exists w\in\Sigma_{uo}^\ast\mbox{ s.t. }x\in\delta(x_0,w)\}$ stands for the (unique) initial state.
We refer the reader to~\cite{Cassandras(2021)} for details on the observer $Obs(G)$.

\begin{remark}\label{re:3.2}
According to Definition~\ref{de:3.1}, we conclude readily that $K$-SSO reduces to standard CSO when $K=0$, i.e., $0$-SSO coincides with standard CSO.
Therefore, $0$-SSO can be determined using the standard observer $Obs(G)$ of $G$.
Specifically, $G$ is $0$-SSO in Definition~\ref{de:3.1} if and only if in its observer $Obs(G)$ there exists no state $q\in X_{obs}$ such that $q\subseteq X_S$, see~\cite{Saboori(2007)} for details.
\end{remark}

To study verification of $K$-SSO ($K\geq 1$) in Definition~\ref{de:3.1}, the state set $X_{obs}$ of $Obs(G)$ is partitioned into three disjoint parts: $X^S_{obs}$, $X^{NS}_{obs}$, and $X^{hyb}_{obs}$, i.e., $X_{obs}=X^S_{obs}\cup X^{NS}_{obs}\cup X^{hyb}_{obs}$, where $X^S_{obs}=\{q\in X_{obs}:q\subseteq X_S\}$, $X^{NS}_{obs}=\{q\in X_{obs}:q\subseteq X_{NS}\}$, and $X^{hyb}_{obs}=\{q\in X_{obs}:q\cap X_S\neq\emptyset\wedge q\cap X_{NS}\neq\emptyset\}$.
Note that the superscript ``\emph{hyb}" of $X^{hyb}_{obs}$ stands for the acronym of ``\emph{hybrid}".

Now, we construct a subautomaton of $G$ called an \emph{initial-secret subautomaton}, denoted by
\begin{equation}\label{eq:3.3}
\hat{G}=(\hat{X},\hat{\Sigma},\hat{\delta},\hat{X}_0),
\end{equation}
which is obtained from $G$ by: 1) replacing its initial state set $X_0$ with $\hat{X}_0=X_S$, and 2) computing the part of $G$ reachable from $\hat{X}_0$ as $\hat{G}$.
Note that, $\hat{G}$ can be computed from $G$ in time linear in the size of $G$.
In particular, when $G$ is deterministic, the time complexity of computing $\hat{G}$ reduces to $\mathcal{O}(|\Sigma||X|)$.

We also construct another subautomaton of $G$ called a \emph{non-secret subautomaton}, denoted by
\begin{equation}\label{eq:3.4}
\tilde{G}=(\tilde{X},\tilde{\Sigma},\tilde{\delta},\tilde{X}_0),
\end{equation}
whose set of initial states is defined as $\tilde{X}_0=\{x\in X:\exists q\in X^{hyb}_{obs}\mbox{ s.t. }x\in q\cap X_{NS}\}$.
And then, in $G$ we delete all secret states and compute the part of $G$ reachable from $\tilde{X}_0$ as $\tilde{G}$.
The time complexity of computing $\tilde{G}$ from $G$ is exponential in the size of $G$.
Note that, when we delete a secret state, all its ingoing and outgoing transitions are also deleted.

Next, we construct the observer of $\tilde{G}$, denoted by $\tilde{G}_{obs}$, which is a minor variant of standard observer $Obs(\tilde{G})$.
The unique difference between constructions is that $\tilde{G}_{obs}$ may have more than one initial state.
Specifically, we construct $\tilde{G}_{obs}=(\tilde{X}_{obs},\tilde{\Sigma}_{obs},\tilde{\delta}_{obs},\tilde{X}_{obs,0})$, where the initial state set is $\tilde{X}_{obs,0}=\{X_{NS}\cap q: \exists q\in X^{hyb}_{obs}\}$.

Based on the above preparation, we propose an information structure called the \emph{concurrent composition} of $\hat{G}$ and $\tilde{G}_{obs}$, which will be used to verify $K$-SSO in Definition~\ref{de:3.1}.

\begin{definition}[Concurrent Composition]\label{de:3.2}
Given a system $G=(X,\Sigma,\delta,X_0)$ and a set $X_{S}\subseteq X$ of secret states, the concurrent composition of $\hat{G}$ and $\tilde{G}_{obs}$ is an NFA
\begin{equation}\label{eq:5}
Cc(\hat{G},\tilde{G}_{obs})=(\hat{X}_{cc},\hat{\Sigma}_{cc},\hat{\delta}_{cc},\hat{X}_{cc,0}),
\end{equation}
where
\begin{itemize}
  \item $\hat{X}_{cc}\subseteq \hat{X}\times 2^{\tilde{X}}$ stands for the set of states;
  \item $\hat{\Sigma}_{cc}=\{(\sigma,\sigma): \sigma\in\hat{\Sigma}_o\}\cup\{(\sigma,\epsilon): \sigma\in\hat{\Sigma}_{uo}\}$ stands for the set of events;
  \item $\hat{\delta}_{cc}: \hat{X}_{cc}\times\hat{\Sigma}_{cc}\rightarrow 2^{\hat{X}_{cc}}$ is the transition function defined as follows: for any state $(x,\tilde{q})\in\hat{X}_{cc}$ and any event $\sigma\in\hat{\Sigma}$,
  \begin{itemize}
  \item [(i)] when $\tilde{q}\neq\emptyset$,\\
    (a) if $\sigma\in\hat{\Sigma}_o$, then
   \begin{equation*}
   \begin{split}
  & \hat{\delta}_{cc}((x,\tilde{q}),(\sigma,\sigma))=\{(x^\prime,\tilde{q}^\prime): x^\prime\in\hat{\delta}(x,\sigma)\wedge\\
  & \tilde{q}^\prime=\tilde{\delta}_{obs}(\tilde{q},\sigma)\mbox{ if }\tilde{\delta}_{obs}(\tilde{q},\sigma) \mbox{ is well-defined}, \tilde{q}^\prime=\emptyset \\
  & \mbox{otherwise}\};
   \end{split}
   \end{equation*}
    (b) if $\sigma\in\hat{\Sigma}_{uo}$, then
   \begin{equation*}
   \hat{\delta}_{cc}((x,\tilde{q}),(\sigma,\epsilon))=\{(x^\prime,\tilde{q}): x^\prime\in\hat{\delta}(x,\sigma)\};
  \end{equation*}
  \item [(ii)] When $\tilde{q}=\emptyset$,\\
     (a) if $\sigma\in\hat{\Sigma}_o$, then
   \begin{equation*}
   \hat{\delta}_{cc}((x,\emptyset),(\sigma,\sigma))=\{(x^\prime,\emptyset): x^\prime\in\hat{\delta}(x,\sigma)\};
   \end{equation*}
     (b) if $\sigma\in\hat{\Sigma}_{uo}$, then
   \begin{equation*}
   \hat{\delta}_{cc}((x,\emptyset),(\sigma,\epsilon))=\{(x^\prime,\emptyset): x^\prime\in\hat{\delta}(x,\sigma)\};
  \end{equation*}
  \end{itemize}
  \item $\hat{X}_{cc,0}=\{(x,\tilde{q}_0): (\exists q\in X^{S}_{obs}\vee X^{hyb}_{obs})[(x\in X_S\cap q)\wedge(\tilde{q}_0=X_{NS}\cap q)]\}\subseteq\hat{X}_0\times\tilde{X}_{obs,0}$ stands for the set of initial states.
\end{itemize}
\end{definition}

\begin{remark}\label{re:3.3}
The concurrent composition $Cc(\hat{G},\tilde{G}_{obs})$ in Definition~\ref{de:3.2} is a variant of that of~\cite{Han(2023)}.
However, there are remarkable differences between them based on their definitions.
Intuitively, $Cc(\hat{G},\tilde{G}_{obs})$ captures that for all $x\in X_S$ and all $s\in\mathcal{L}(G,x)$ whether there exists an observation $\alpha\in\mathcal{L}(\tilde{G}_{obs},\tilde{q}_0)$ such that $P(s)=\alpha$ and the initial state $\tilde{q}_0\in\tilde{X}_{obs,0}$ of observer $\tilde{G}_{obs}$ satisfies $\tilde{q}_0=q\backslash X_S$ for some $q$ satisfying $x\in q\in X^{hyb}_{obs}$.
In other words, in $Cc(\hat{G},\tilde{G}_{obs})$ if there exists a state of the form $(\cdot,\emptyset)$, then there exists a non-negative integer $K$ such that $G$ is not $K$-SSO, and vice versa.
In addition, for a sequence $e\in\mathcal{L}(Cc(\hat{G},\tilde{G}_{obs}))$, we use the notations $e(L)$ and $e(R)$ to denote its left and right components, respectively.
Further, $P(e)$ denotes $P(e(L))$ or $e(R)$ since $P(e(L))=P(e(R))=e(R)$.
\end{remark}

\subsection{Verification for \texorpdfstring{$K$}{K}-SSO in Definition~\ref{de:3.1}}\label{subsec3.3}

In this subsection, we are ready to present the main result on the verification of $K$-SSO in Definition~\ref{de:3.1} using the proposed concurrent composition $Cc(\hat{G},\tilde{G}_{obs})$.

\begin{theorem}\label{th:3.1}
Given a system $G=(X,\Sigma,\delta,X_0)$, a projection map $P$ w.r.t. a set $\Sigma_o$ of observable events, a set $X_{S}$ of secret states, and a non-negative integer $K$,
let $Cc(\hat{G},\tilde{G}_{obs})$ be the corresponding concurrent composition.
$G$ is $K$-SSO in Definition~\ref{de:3.1} w.r.t. $\Sigma_o$ and $X_S$ if and only if there exists no state of the form $(\cdot,\emptyset)$ in $Cc(\hat{G},\tilde{G}_{obs})$ that is observationally reachable from $\hat{X}_{cc,0}$ within $K$ steps.
\end{theorem}

\begin{proof}
$(\Rightarrow)$ By contrapositive, assume that there exists a state $(x,\emptyset)$ in $Cc(\hat{G},\tilde{G}_{obs})$ that is $k$-step observationally reachable from $\hat{X}_{cc,0}$, where $k\leq K$.
By the construction of $Cc(\hat{G},\tilde{G}_{obs})$, there exists an initial state $(x_1,\tilde{q}_0)\in\hat{X}_{cc,0}$ and a sequence $e\in\mathcal{L}(Cc(\hat{G},\tilde{G}_{obs}),(x_1,\tilde{q}_0))$ with $|P(e)|=k$ such that $(x,\emptyset)\in\hat{\delta}_{cc}((x_1,\tilde{q}_0),e)$, where $x_1\in X_S\cap q$ and $\tilde{q}_0=X_{NS}\cap q$ for some $q\in X^{S}_{obs}\vee X^{hyb}_{obs}$.
There are two cases:
When $\tilde{q}_0=\emptyset$, since $(x_1,\emptyset)$ is an initial state of $Cc(\hat{G},\tilde{G}_{obs})$, we conclude that $G$ is not $0$-SSO based on the construction of $Cc(\hat{G},\tilde{G}_{obs})$.
Hence, it is also not $K$-SSO w.r.t. $\Sigma_o$ and $X_S$.
When $\tilde{q}_0\neq\emptyset$, by the constructions of $\hat{G}$ and $\tilde{G}_{obs}$, we further have that: 1) $x\in\hat{\delta}(x_1,e(L))$, and 2) $\tilde{\delta}_{obs}(\tilde{q}_0,e(R))$ is not well-defined.
Item 1) means $x\in\delta(x_1,e(L))$.
Item 2) means, by the construction of $\tilde{G}$, that for all $x_1^\prime\in\tilde{q}_0$ and all $s_2^\prime\in\tilde{\Sigma}^\ast$ with $P(s_2^\prime)=e(R)$, it holds $\tilde{\delta}(x_1^\prime,s_2^\prime)=\emptyset$.
Therefore, there exists no non-secret run in $G$ starting at $x_1^\prime$ such that its observation is $e(R)$.
On the other hand, we conclude $\{x_1\}\cup\tilde{q}_0\subseteq q\in X^{hyb}_{obs}$ based on the fact of $\tilde{q}_0\neq\emptyset$.
Thus by the construction of $Obs(G)$, we conclude that: 1) in $G$ there exists an initial state $x_0\in X_0$ and a sequence $s_1\in\mathcal{L}(G,x_0)$ such that $x_1\in\delta(x_0,s_1)$, and 2) for each $x_1^\prime\in\tilde{q}_0$, there exists an initial state $x_0^\prime\in X_0$ and a sequence $s_1^\prime\in\mathcal{L}(G,x_0^\prime)$ with $P(s_1^\prime)=P(s_1)$ such that $x_1^\prime\in\delta(x_0^\prime,s_1^\prime)$.
Taken together, for the run $x_0\stackrel{s_1}{\rightarrow}x_1\stackrel{e(L)}{\rightarrow}x$ generated by $G$ with $|P(e(L))|=k\leq K$, there exists no run $x^\prime_0\stackrel{s_1^\prime}{\rightarrow}{x^\prime_1}\stackrel{s_2^\prime}{\rightarrow}{x_2^\prime}$ such that ${x^\prime_1}\stackrel{s_2^\prime}{\rightarrow}{x^\prime_2}$ is non-secret, where $P(s_1^\prime)=P(s_1)$ and $P(s_2^\prime)=P(s_2)$.
By Definition~\ref{de:3.1}, $G$ is not $K$-SSO w.r.t. $\Sigma_o$ and $X_S$.

$(\Leftarrow)$ Also by contrapositive, assume that $G$ is not $K$-SSO w.r.t. $\Sigma_o$ and $X_S$.
By Definition~\ref{de:3.1}, we conclude that in $G$: 1) there exists a run $x_0\stackrel{s_1}{\rightarrow}x_1\stackrel{s_2}{\rightarrow}x_2$, where $x_0\in X_0$, $x_1\in X_S$, and $|P(s_2)|=k\leq K$,
and 2) there exists no run $x^\prime_0\stackrel{s_1^\prime}{\rightarrow}{x^\prime_1}\stackrel{s_2^\prime}{\rightarrow}{x^\prime_2}$ such that ${x^\prime_1}\stackrel{s_2^\prime}{\rightarrow}{x^\prime_2}$ is non-secret, where $x^\prime_0\in X_0$, $P(s_1^\prime)=P(s_1)$, and $P(s_2^\prime)=P(s_2)$.
Item 1) means $x_1\in\delta(x_0,s_1)$ and $x_2\in\hat{\delta}(x_1,s_2)$ based on the construction of $\hat{G}$.
By the construction of $\tilde{G}$, item 2) means $\tilde{\delta}(x_1^{\prime},s_2^\prime)=\emptyset$ for all $x_1^\prime\in\tilde{q}_0=X_{NS}\cap q$, where $x_1^\prime\in q\in X^{S}_{obs}\vee X^{hyb}_{obs}$, $s_2^\prime\in\tilde{\Sigma}^\ast$, and $P(s_2^\prime)=P(s_2)$\footnote{Note that, when $q\in X^{S}_{obs}$, we have $\tilde{q}_0=X_{NS}\cap q=\emptyset$. For this case, we conclude that $(x_1,\emptyset)$ is an initial state of $Cc(\hat{G},\tilde{G}_{obs})$. Thus we can directly obtain $(x_2,\emptyset)\in\hat{\delta}_{cc}((x_1,\emptyset),e)$.}.
By the construction of $\tilde{G}_{obs}$, we further obtain $\tilde{\delta}_{obs}(\tilde{q}_0,P(s_2^\prime))=\emptyset$.
Then by the definition of $Cc(\hat{G},\tilde{G}_{obs})$, we conclude that there exists a sequence $e\in\mathcal{L}(Cc(\hat{G},\tilde{G}_{obs}))$ with $e(L)=s_2$ and $e(R)=P(s_2^\prime)$ such that $(x_2,\emptyset)\in\hat{\delta}_{cc}((x_1,\tilde{q}_0),e)$.
Since $|P(e)|=|P(e(L))|=|P(s_2)|=k$, state $(x_2,\emptyset)$ is $k$-step observationally reachable from $\hat{X}_{cc,0}$ in $Cc(\hat{G},\tilde{G}_{obs})$.
\end{proof}

Based on Theorem~\ref{th:3.1}, the verification procedure for $K$-SSO in Definition~\ref{de:3.1} can be summed up as the following Algorithm~\ref{algo:1}.
\begin{algorithm}
\caption{Verification of $K$-SSO in Definition~\ref{de:3.1}}\label{algo:1}
\begin{algorithmic}[1]
\REQUIRE System $G=(X,\Sigma,\delta,X_0)$, set $\Sigma_o$ of observable events, set $X_{S}$ of secret states, and non-negative integer $K$
\ENSURE ``Yes'' if $G$ is $K$-SSO w.r.t. $\Sigma_o$ and $X_S$, ``No" otherwise
\STATE Compute the observer $Obs(G)$ of $G$
\IF{there exists a reachable state $q\in X_{obs}$ such that $q\subseteq X_S$}
\STATE $G$ is not $0$-SSO w.r.t. $\Sigma_o$ and $X_S$
\RETURN ``No"
\STOP
\ELSE
\STATE Construct the initial-secret subautomaton $\hat{G}$ of $G$
\STATE Construct the non-secret subautomaton $\tilde{G}$ of $G$
\STATE Compute the observer $\tilde{G}_{obs}$ of $\tilde{G}$
\STATE Compute the corresponding $Cc(\hat{G},\tilde{G}_{obs})$
\STATE Use the ``Breadth-First Search Algorithm" in~\cite{Cormen(2009)} to find whether there exists a state of form $(\cdot,\emptyset)$ in $Cc(\hat{G},\tilde{G}_{obs})$ that is observationally reachable from $\hat{X}_{cc,0}$ within $K$ steps
\IF{such a state $(\cdot,\emptyset)$ in $Cc(\hat{G},\tilde{G}_{obs})$ exists}
\RETURN ``No"
\STOP
\ELSE
\RETURN ``Yes"
\STOP
\ENDIF
\ENDIF
\end{algorithmic}
\end{algorithm}

\begin{remark}\label{re:3.4}
We discuss the time complexity of using the proposed concurrent composition $Cc(\hat{G},\tilde{G}_{obs})$ to verify $K$-SSO in Definition~\ref{de:3.1}.
Specifically, computing $Obs(G)$, $\hat{G}$, $\tilde{G}$, $\tilde{G}_{obs}$, and $Cc(\hat{G},\tilde{G}_{obs})$ take time $\mathcal{O}(|\Sigma_o||\Sigma_{uo}||X|^{2}2^{|X|})$, $\mathcal{O}(|\Sigma||X|^2)$, $\mathcal{O}(|\Sigma||X|^2+|X|2^{|X|})$, $\mathcal{O}(|\Sigma_o||\Sigma_{uo}||X|^{2}2^{|X|})$, and $\mathcal{O}(|\Sigma||X|^{2}2^{|X|})$, respectively.
As a consequence, the overall (worst-case) time complexity of verifying $K$-SSO in Definition~\ref{de:3.1} using Algorithm~\ref{algo:1} is $\mathcal{O}((|\Sigma_o||\Sigma_{uo}|+|\Sigma|)|X|^{2}2^{|X|})$, 
where $|X|$ (resp., $|\Sigma|$, $|\Sigma_o|$, $|\Sigma_{uo}|$) is the number of states (resp., events, observable events, unobservable events) of system $G$. 
Note that, although Algorithm~\ref{algo:1} for verifying $K$-SSO in Definition~\ref{de:3.1} depends on $K$, the construction of the concurrent composition $Cc(\hat{G},\tilde{G}_{obs})$ does not depend on $K$.
\end{remark}

\begin{example}[\cite{Ma(2021)}]\label{ex:3.1}
Let us consider the system $G$ shown in Fig.~\ref{Fig1} in which $\Sigma_o=\{a,b,c\}$, $\Sigma_{uo}=\{u\}$, $X_0=\{0\}$, and $X_S=\{5,7\}$.
By applying Algorithm~\ref{algo:1}, we obtain the corresponding $Obs(G)$, $\hat{G}$, $\tilde{G}$, $\tilde{G}_{obs}$, and $Cc(\hat{G},\tilde{G}_{obs})$, which are depicted in~Fig.~\ref{Fig2}(a-e), respectively.
In $Cc(\hat{G},\tilde{G}_{obs})$ there exists a state $(8,\emptyset)$ that is $2$-step observationally reachable from the initial-state $(7,\{1,2,3,4\})$.
Then by Theorem~\ref{th:3.1}, $G$ is $1$-SSO in Definition~\ref{de:3.1}, but not $K$-SSO for any $K>1$.
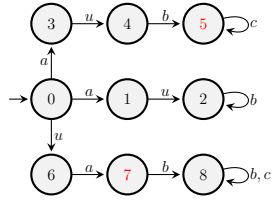
\begin{figure}[!ht]
  \centering
  	  \begin{tikzpicture}[>=stealth',shorten >=1pt,auto,node distance=2.0 cm, scale = 0.5, transform shape,
	>=stealth,inner sep=2pt]

	\node[initial, initial where = left, state] (0) {$0$};
	\node[state] (1) [right of =0] {$1$};
	\node[state] (2) [right of =1] {$2$};
	\node[state] (3) [above of =0] {$3$};
	\node[state] (4) [right of =3] {$4$};
	\node[state] (5) [right of =4] {$\red 5$};
	\node[state] (6) [below of =0] {$6$};
	\node[state] (7) [right of =6] {$\red 7$};
	\node[state] (8) [right of =7] {$8$};

	\path [->]
	(0) edge node [above, sloped] {$a$} (1)
	(1) edge node [above, sloped] {$u$} (2)
	(3) edge node [above, sloped] {$u$} (4)
	(4) edge node [above, sloped] {$b$} (5)
	(6) edge node [above, sloped] {$a$} (7)
	(7) edge node [above, sloped] {$b$} (8)
	(0) edge node {$a$} (3)
	(0) edge node {$u$} (6)
	(5) edge [loop right] node {$c$} (5)
	(2) edge [loop right] node {$b$} (2)
	(8) edge [loop right] node {$b,c$} (8)
	;

        \end{tikzpicture}
  \caption{The automaton $G$ considered in Example~\ref{ex:3.1}.}
  \label{Fig1}
\end{figure}

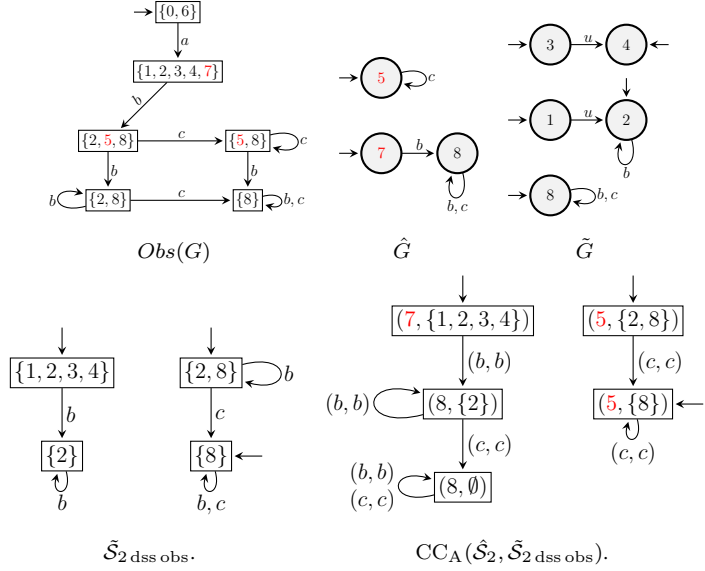
\begin{figure}[!ht]
  \centering
  \subcaptionbox*{\scriptsize$Obs(G)$}{

	  		  \begin{tikzpicture}[>=stealth',shorten >=1pt,auto,node distance=2.6 cm, scale = 0.5, transform shape,
	>=stealth,inner sep=2pt]

	\node[initial, initial where = left, rectangle state] (06) {$\{0,6\}$};
	\node[rectangle state] (12347) [below = 1cm of 06] {$\{1,2,3,4,{\red 7}\}$};
	\node[rectangle state] (258) [below left of = 12347] {$\{2,{\red 5},8\}$};
	\node[rectangle state] (58) [below right of = 12347] {$\{{\red 5},8\}$};
	\node[rectangle state] (28) [below = 1cm of 258] {$\{2,8\}$};
	\node[rectangle state] (8) [below = 1cm of 58] {$\{8\}$};
	
	\path [->]
	(06) edge node {$a$} (12347)
	(12347) edge node [above, sloped] {$b$} (258)
	(258) edge node [above, sloped] {$c$} (58)
	(58) edge node {$b$} (8)
	(58) edge [loop right] node {$c$} (58)
	(8) edge [loop right] node {$b,c$} (8)
	(258) edge node {$b$} (28)
	(28) edge node [above, sloped] {$c$} (8)
	(28) edge [loop left] node {$b$} (28)
	;

        \end{tikzpicture}
	}
	\subcaptionbox*{\scriptsize$\hat G$}{
	\begin{tikzpicture}[>=stealth',shorten >=1pt,auto,node distance=2.0 cm, scale = 0.5, transform shape,
	>=stealth,inner sep=2pt]

	\node[initial, initial where = left, state] (5) {$\red 5$};
	\node[initial, initial where = left, state] (7) [below of =5] {$\red 7$};
	\node[state] (8) [right of =7] {$8$};

	\path [->]
	(7) edge node [above, sloped] {$b$} (8)
	(5) edge [loop right] node {$c$} (5)
	(8) edge [loop below] node {$b,c$} (8)
	;

        \end{tikzpicture}
	}
	\subcaptionbox*{\scriptsize$\tilde G$}{
	\begin{tikzpicture}[>=stealth',shorten >=1pt,auto,node distance=2.0 cm, scale = 0.5, transform shape,
	>=stealth,inner sep=2pt]

	\node[initial, initial where = left, state] (1) {$1$};
	\node[initial, initial where = above, state] (2) [right of =1] {$2$};
	\node[initial, initial where = left, state] (3) [above of =1] {$3$};
	\node[initial, initial where = right, state] (4) [right of =3] {$4$};
	\node[initial, initial where = left, state] (8) [below of =1] {$8$};

	\path [->]
	(1) edge node [above, sloped] {$u$} (2)
	(3) edge node [above, sloped] {$u$} (4)
	(2) edge [loop below] node {$b$} (2)
	(8) edge [loop right] node {$b,c$} (8)
	;

        \end{tikzpicture}
	 }

	 \subcaptionbox*{\scriptsize$\tilde\Scal_{2\dss\obs}$.}{
	 \begin{tikzpicture}[>=stealth',shorten >=1pt,auto,node distance=2.8 cm, scale = 0.7, transform shape,
	>=stealth,inner sep=2pt]

	\node[initial, initial where = above, rectangle state] (1234) {$\{1,2,3,4\}$};
	\node[initial, initial where = above, rectangle state] (28) [right of =1234] {$\{2,8\}$};
	\node[rectangle state] (2) [below = 1cm of 1234] {$\{2\}$};
	\node[initial, initial where = right, rectangle state] (8) [below = 1cm of 28] {$\{8\}$};

	\path [->]
	(1234) edge node {$b$} (2)
	(28) edge node {$c$} (8)
	(28) edge [loop right] node {$b$} (28)
	(2) edge [loop below] node {$b$} (2)
	(8) edge [loop below] node {$b,c$} (8)
	;
    \end{tikzpicture}
	 }
	 \subcaptionbox*{\scriptsize$\CCa(\hat\Scal_2,\tilde\Scal_{2\dss\obs})$.}{
	  \begin{tikzpicture}[>=stealth',shorten >=1pt,auto,node distance=3.2 cm, scale = 0.7, transform shape,
	>=stealth,inner sep=2pt]

	\node[initial, initial where = above, rectangle state] (7-1234) {$({\red 7},\{1,2,3,4\})$};
	\node[initial, initial where = above, rectangle state] (5-28) [right of = 7-1234] {$({\red 5},\{2,8\})$};
	\node[rectangle state] (8-2) [below = 1cm of 7-1234] {$(8,\{2\})$};
	\node[initial, initial where = right, rectangle state] (5-8) [below = 1cm of 5-28] {$({\red 5},\{8\})$};
	\node[rectangle state] (8-phi) [below = 1cm of 8-2] {$(8,\emptyset)$};

	\path [->]
	(7-1234) edge node {$(b,b)$} (8-2)
	(5-28) edge node {$(c,c)$} (5-8)
	(8-2) edge node {$(c,c)$} (8-phi)
	(8-2) edge [loop left] node {$(b,b)$} (8-2)
	(5-8) edge [loop below] node {$(c,c)$} (508)
	(8-phi) edge [loop left] node {$\begin{matrix}(b,b)\\(c,c)\end{matrix}$} (8-phi)
	;

        \end{tikzpicture}
	  }
  \caption{The constructed automata $Obs(G)$, $\hat{G}$, $\tilde{G}$, $\tilde{G}_{obs}$, and $Cc(\hat{G},\tilde{G}_{obs})$ from the automaton $G$ shown in Fig.~\ref{Fig1}.}
  \label{Fig2}
\end{figure}

\end{example}

\subsection{An upper bound on \texorpdfstring{$K$}{K} in \texorpdfstring{$K$}{K}-SSO in Definition~\ref{de:3.1}}\label{subsec3.4}

In~\cite{Ma(2021)}, the authors proposed the notion of Inf-SSO and investigated its verification.
Recently, an improved algorithm for verifying Inf-SSO has been designed in our previous work~\cite{Han(2023)}.
Moreover, the authors in~\cite{Ma(2021)} used the so-called $K$-step recognizer to derive an upper bound on $K$ in $K$-SSO, i.e., $|X|(2^{|X|}-1)$.
In fact, this upper bound is conservative, which has been improved to $|\hat{X}|2^{|X\backslash X_S|}-1$ in our previous work~\cite{Zhang(2023a)}.

Now we are ready to derive an upper bound on $K$ in $K$-SSO in Definition~\ref{de:3.1}.
Specifically, by Definition~\ref{de:3.1}, we conclude readily that if a system $G$ is $K$-SSO w.r.t. $\Sigma_o$ and $X_S$, then it is also $K^\prime$-SSO for any $K^\prime\leq K$.
Conversely, if $G$ is not $K$-SSO with $K> |\hat{X}|2^{|X\backslash X_S|}-1$, by Definition~\ref{de:3.1},
we know that in $G$: 1) there exists a run $x_0\stackrel{s_1}{\rightarrow}x_1\stackrel{s_2}{\rightarrow}x_2$, where $x_0\in X_0$, $x_1\in X_S$, and $|P(s_2)|=k\leq K$,
and 2) there exists no run $x^\prime_0\stackrel{s_1^\prime}{\rightarrow}{x^\prime_1}\stackrel{s_2^\prime}{\rightarrow}{x^\prime_2}$ such that ${x^\prime_1}\stackrel{s_2^\prime}{\rightarrow}{x^\prime_2}$ is non-secret, where $x^\prime_0\in X_0$, $P(s_1^\prime)=P(s_1)$, and $P(s_2^\prime)=P(s_2)$.
Thus, we conclude that state $(x_2,\emptyset)$ in $Cc(\hat{G},\tilde{G}_{obs})$ is $k$-step observationally reachable from $\hat{X}_{cc,0}$ (see ``$\Leftarrow$" part in the proof of Theorem~\ref{th:3.1}).
Since $Cc(\hat{G},\tilde{G}_{obs})$ has at most $|\hat{X}|2^{|X\backslash X_S|}$ states, there exists a $k^\prime\leq|\hat{X}|2^{|X\backslash X_S|}-1$ such that $(x_2,\emptyset)$ is $k^\prime$-step observationally reachable from $\hat{X}_{cc,0}$.
Then by Theorem~\ref{th:3.1}, $G$ is not $k^\prime$-SSO.
Therefore, it is not $(|\hat{X}|2^{|X\backslash X_S|}-1)$-SSO.
This means that the result of determining $K$-SSO Definition~\ref{de:3.1} using Theorem~\ref{th:3.1} does not depend on $K$ when $K\geq|\hat{X}|2^{|X\backslash X_S|}-1$.
Therefore, the upper bound on $K$ in $K$-SSO in Definition~\ref{de:3.1} is also $|\hat{X}|2^{|X\backslash X_S|}-1$, which directly results in Theorem~\ref{th:3.2} below.

\begin{theorem}\label{th:3.2}
Given a system $G=(X,\Sigma,\delta,X_0)$, a projection map $P$ w.r.t. a set $\Sigma_o$ of observable events, a set $X_{S}$ of secret states, and a non-negative integer $K$, 
$G$ is $K$-SSO in Definition~\ref{de:3.1} w.r.t. $\Sigma_o$ and $X_S$ if and only if it is min$\{K, |\hat{X}|2^{|X\backslash X_S|}-1\}$-SSO in Definition~\ref{de:3.1} w.r.t. $\Sigma_o$ and $X_S$.
\end{theorem}

\begin{remark}\label{re:3.5}
The authors in~\cite{Yin(2017)} derived an upper bound (i.e., $2^{|X|}-2$) on $K$ in the standard $K$-SO to show that both the standard $K$-SO and Inf-SO notions are equivalent based on this upper bound.
Similar to the standard $K$-SO, there is also an upper bound on $K$ for $K$-SSO defined in~\cite{Falcone(2015)} (resp., defined in~\cite{Zhang(2023b)}).
The upper bound is $|X|(2^{|X\backslash X_S|}-1)$.
In contrast, there exists no such an upper bound on $K$ in $K$-SSO in Definition~\ref{de:3.1} that establishes the equivalent relationship between $K$-SSO and Inf-SSO.
Let us consider the system $G$ shown in Fig.~\ref{Fig3}.
By Definitions~\ref{de:3.1} and~\ref{de:4.2} below, we conclude readily that $G$ is $K$-SSO in Definition~\ref{de:3.1} for any given non-negative integer $K$, but not Inf-SSO.
\begin{figure}[!ht]
  \centering
  \begin{tikzpicture}[>=stealth',shorten >=1pt,auto,node distance=2.0 cm, scale = 0.7, transform shape,
	>=stealth,inner sep=2pt]

	\node[initial, initial where = left, state] (0) {$0$};
	\node[state, right of = 0] (1) {$\red 1$};
	\node[state, above right of = 1] (2) {$2$};
	\node[state, below right of = 1] (3) {$3$};

	\path [->]
	(0) edge node [above, sloped] {$u$} (1)
	(1) edge node [above, sloped] {$a$} (2)
	(2) edge [loop right] node {$b$} (2)
	(1) edge node [above, sloped] {$u$} (3)
	(3) edge node {$a$} (2)
	;

  \end{tikzpicture}
  \caption{A automaton $G$ in which $\Sigma_o=\{a,b\}$, $X_0=\{0\}$, and $X_S=\{1\}$.}
  \label{Fig3}
\end{figure}
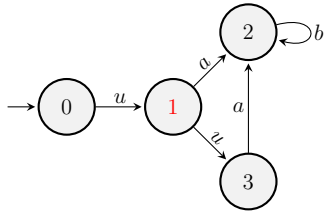
\end{remark}

\section{Enforcement of strong SBO}\label{sec4}

In this section, we focus on how to use the proposed two concurrent-composition structures to enforce the strong SBO notions: $K$-SSO, SCSO, SISO, and Inf-SSO.
We propose a novel opacity-enforcement mechanism, which chooses transitions labeled by controllable events to disable before an original system starts to run for the sake of cutting off all the runs violating the corresponding strong SBO definition (if possible).
Compared with the classical supervisory control framework that dynamically chooses controllable events to disable according to the observed output sequences
(see, e.g., \cite{Dubreil(2010),Saboori(2012b),Yin(2016),Tong(2018),Moulton(2022)}), the proposed opacity-enforcement mechanism has better scalability and wider implementability since it removes the limitation of \emph{observability condition} in supervisory control framework (i.e., all sequences generated by an original system that have the same projection require the same control action) and other assumptions such as the containment relation between $\Sigma_o$ and $\Sigma_c$, etc.
Note that, the two aforementioned opacity-enforcement approaches have the same prerequisite: the structure of a manmade system is known and partially changeable in advance.

This section includes three parts below: Subsection~\ref{subsec4.1} is \emph{problem formulation} for enforcing the notions of strong SBO.
In subsection~\ref{subsec4.2}, we use the proposed concurrent-composition $Cc(\hat{G},\tilde{G}_{obs})$ constructed in Section~\ref{sec3} to design an efficient algorithm for enforcing $K$-SSO in Definition~\ref{de:3.1}.
The enforcement algorithms of SCSO, SISO, and Inf-SSO using our previously-proposed structure $Cc(G, Obs(G_{dss}))$ in~\cite{Han(2023)} are presented in subsection~\ref{subsec4.3}.

\subsection{Problem formulation}\label{subsec4.1}

By convention, when we discuss the problem of enforcing opacity for a given system $G=(X,\Sigma,\delta,$ $X_0)$, its event set $\Sigma$ is partitioned into two disjoint subsets, i.e., $\Sigma=\Sigma_c\cup\Sigma_{uc}$, where $\Sigma_c$ is the set of controllable events and $\Sigma_{uc}$ is the set of uncontrollable events.
In our previous work~\cite{Han(2023)} and the Section~\ref{sec3} of this paper, we proposed two concurrent-composition structures $Cc(G, Obs(G_{dss}))$ and $Cc(\hat{G},\tilde{G}_{obs})$ to verify SCSO, SISO, Inf-SSO, and $K$-SSO, respectively, which is currently the most efficient opacity verification algorithms.
In the following, we will use these two structures to enforce them under the proposed framework.
We name it \emph{strong SBO enforcement problem} (SSBOEP).

To formulate SSBOEP, we use the notation $\mathcal{T}^G$ to denote the set consisting of all transitions of system $G$.
In particular, a transition $x\stackrel{\sigma}{\rightarrow}y\in\mathcal{T}^G$ is called \emph{controllable} if $\sigma\in\Sigma_c$.
A \emph{controllable run} in $G$ is a run that contains a controllable transition.
We denote the set of all controllable transitions of $G$ by $\mathcal{T}_c^G$.
Note that, if all transitions in $\mathcal{T}^G$ are uncontrollable, then it holds $\mathcal{T}_c^G=\emptyset$.
It is traditionally assumed that one can only disable controllable transitions to restrict the behavior of $G$, while its uncontrollable transitions can never be disabled.

\begin{problem}\label{problem:1}
Given a system $G=(X,\Sigma,\delta,X_0)$, a projection map $P$ w.r.t. a set $\Sigma_o$ of observable events, a set $\Sigma_c$ of controllable events, and a set $X_{S}\subseteq X$ of secret states,
compute a subset $\mathcal{E}_c\subseteq\mathcal{T}_c^G$ of controllable transitions (if one exists) such that the subsystem $G_{\backslash\mathcal{E}_c}$ obtained from system $G$ by disabling all transitions in $\mathcal{E}_c$ is $K$-SSO in Definition~\ref{de:3.1} for a given non-negative integer $K$ (resp., SCSO, SISO, and Inf-SSO) w.r.t. $\Sigma_o$ and $X_S$.
\end{problem}

Note that, Problem~\ref{problem:1} may have no solution due to system settings on uncontrollable events and secret states.
However, Problem~\ref{problem:1} may have more than one solution.
In this paper, we only consider how to find a valid solution of Problem~\ref{problem:1} (if exists) by choosing a number of controllable transitions to disable in order to cut off all runs in the corresponding concurrent composition that violate a strong SBO notion of interest.
In particular, when system $G$ is opaque\footnote{In this paper, the usage of the word ``opacity/opaque"  hereafter means ``one of all strong SBO notions", which depends on the context.}, such a solution of Problem~\ref{problem:1} is equal to the empty set (i.e., $\mathcal{E}_c=\emptyset$).
In the following two subsections, we design efficient algorithms to check the existence of a solution to Problem~\ref{problem:1} and compute one if one exists.

\subsection{Enforcement of \texorpdfstring{$K$}{K}-SSO in Definition~\ref{de:3.1}}\label{subsec4.2}

In this subsection, we design an efficient algorithm in the proposed opacity-enforcement framework to enforce $K$-SSO in Definition~\ref{de:3.1} using the concurrent-composition structures $Cc(\hat{G},\tilde{G}_{obs})$ (proposed in Section~\ref{sec3}) and $Cc(G, Obs(G))$.
It is noted that $Cc(G, Obs(G))$ is a slight variant of the parallel composition $G||Obs(G)$ in the current special case that the events of $Obs(G)$ are also events of $G$ and $Obs(G)$ does not have unobservable events.
Specifically, $Cc(G, Obs(G))$ can be obtained from $G||Obs(G)$ by extending the scalar events of $G||Obs(G)$ to the events of vector form by doubling events.
However, in more general cases, these two composition operations are remarkably different, see~\cite{Zhang(2023a)} for details.
Note that, the parallel composition $G||Obs(G)$ has been employed to enforce standard CSO under the supervisory control framework with a number of assumptions~\cite{Tong(2018),Moulton(2022)}.
Now we recall the properties of $Cc(G, Obs(G))$/$G||Obs(G)$.
Details see~\cite{Tong(2018)} and/or~\cite{Moulton(2022)}.

\begin{proposition}[\cite{Moulton(2022)}]\label{pro:4.1}
Given a system $G=(X,\Sigma,\delta,$ $X_0)$, a projection map $P$ w.r.t. a set $\Sigma_o$ of observable events, and a set $X_{S}$ of secret states, let $Obs(G)$ be the observer of $G$.
Then
\begin{enumerate}
  \item For any $q\in X_{obs}$ and any $x\in q$, $(x, q)$ is a state in $Cc(G, Obs(G))$;
  \item $(x, q)$ is a state in $Cc(G, Obs(G))$ if and only if there exists a string $s\in\mathcal{L}(G)$ that leads to state $x$ in $G$ and whose projection leads to estimate $q$ in $Obs(G)$;
  \item $G$ is standard CSO w.r.t. $\Sigma_o$ and $X_{S}$ if and only if there exists no state $(x, q)$ in $Cc(G, Obs(G))$ such that $q\subseteq X_S$.
\end{enumerate}
\end{proposition}

The following result establishes the relationship between the initial states in $Cc(\hat{G},\tilde{G}_{obs})$ and the corresponding states in $Cc(G, Obs(G))$.

\begin{proposition}\label{pro:4.2}
Given a system $G=(X,\Sigma,\delta,X_0)$ and a set $X_{S}$ of secret states, let $Cc(G, Obs(G))$ and $Cc(\hat{G},\tilde{G}_{obs})$ be two corresponding concurrent compositions.
Then for any initial state $(x,\tilde{q}_0)\in\hat{X}_{cc,0}$ in $Cc(\hat{G},\tilde{G}_{obs})$, there exists at least one corresponding state $(x, q)$ in $Cc(G, Obs(G))$ such that $\tilde{q}_0=q\setminus X_S$.
\end{proposition}

\begin{proof}
By definitions of $Cc(G,Obs(G))$ and $Cc(\hat{G},\tilde{G}_{obs})$, this proof trivially holds.
\end{proof}

By Proposition~\ref{pro:4.2}, we conclude that when an initial state $(x,\tilde{q}_0)\in\hat{X}_{cc,0}$ of $Cc(\hat{G},\tilde{G}_{obs})$ violates $K$-SSO in Definition~\ref{de:3.1}, in order to prevent it from being generated, in $Cc(G, Obs(G))$ we need to cut off all runs ended with state $(x, q)$ by disabling the corresponding controllable transitions, where $\tilde{q}_0=q\setminus X_S$.
Based on the above preparation, we design Algorithm~\ref{algo:2} to enforce $K$-SSO in Definition~\ref{de:3.1}.

\begin{algorithm}
\caption{Enforcement of $K$-SSO in Definition~\ref{de:3.1}}\label{algo:2}
\begin{algorithmic}[1]
\REQUIRE System $G=(X,\Sigma,\delta,X_0)$, set $\Sigma_o\subseteq\Sigma$ of observable events, set $\Sigma_c\subseteq\Sigma$ of controllable events, set $X_{S}\subseteq X$ of secret states, and non-negative integer $K$
\ENSURE ``Yes" if $G$ can be enforced $K$-SSO in Definition~\ref{de:3.1}, ``No" otherwise; in case of ``Yes", a subset $\mathcal{E}_c\subseteq\mathcal{T}_c^G$ of controllable transitions such that subsystem $G_{\backslash\mathcal{E}_c}$ is $K$-SSO in Definition~\ref{de:3.1}
\STATE Compute $Cc(\hat{G},\tilde{G}_{obs})$ and $Cc(G,Obs(G))$
\STATE Initialize $\mathcal{E}_c=\varnothing$
\IF{$G$ is $K$-SSO in Definition~\ref{de:3.1}}
\RETURN ``Yes" and ``$\mathcal{E}_c=\varnothing$"
\STOP
\ELSE
\WHILE{there exists a state of the form $(\cdot,\emptyset)$ in $Cc(\hat{G},\tilde{G}_{obs})$ that are observationally reachable from $\hat{X}_{cc,0}$ within $K$ steps}
\STATE Collect all states of the form $(\cdot,\emptyset)$ in $Cc(\hat{G},\tilde{G}_{obs})$ that are observationally reachable from $\hat{X}_{cc,0}$ within $K$ steps as a set $\Theta$
\FOR{each state of the form $(\cdot,\emptyset)$ in $\Theta$}
\IF{there exists an uncontrollable leaking-secret run $(x,\tilde{q}_0)\stackrel{e}{\rightarrow}(\cdot,\emptyset)$ generated by $Cc(\hat{G},\tilde{G}_{obs})$, where $|P(e)|\leq K$}
\STATE Mark all states $(x, q)$ in $Cc(G, Obs(G))$ that correspond to state $(x,\tilde{q}_0)$, where $\tilde{q}_0=q\setminus X_S$

\IF{there exists an uncontrollable run from the initial-state set of $Cc(G,Obs(G))$ to state $(x,q)$}
\RETURN ``No"
\STOP
\ELSE
\STATE Compute the set $\Lambda$ consisting of all last controllable transitions of all controllable runs from $\hat{X}_{cc,0}$ to $(\cdot,\emptyset)$ in $\Theta$ of observational length at most $K$ and all last controllable transitions of all controllable runs $(x_0,q_0)\stackrel{\tilde{e}}{\rightarrow}(x, q)$ (where $x\in X_S$) generated by $Cc(G,Obs(G))$ that are predecessor runs of all uncontrollable runs (if exist) from $\hat{X}_{cc,0}$ to $(\cdot,\emptyset)$ in $\Theta$ of observational length at most $K$ of $Cc(\hat{G},\tilde{G}_{obs})$
\FOR{each controllable transition $(x_1,q_1)\stackrel{{(\sigma,\tilde{\sigma}})}{\longrightarrow}(x_2, q_2)$ in $\Lambda$, where $\sigma\in\Sigma_c$}\label{line1}
\STATE Cut off $(x_1,q_1)\stackrel{{(\sigma,\tilde{\sigma}})}{\longrightarrow}(x_2,q_2)$ by disabling its left component $x_1\stackrel{\sigma}{\rightarrow}x_2$ in $G$, add $x_1\stackrel{\sigma}{\rightarrow}x_2$ to $\mathcal{E}_c$
\ENDFOR
\ENDIF
\ENDIF
\ENDFOR
\STATE Update $Ac(Cc(\hat{G},\tilde{G}_{obs}))$ and $Ac(Cc(G,Obs(G)))$
\ENDWHILE
\ENDIF
\end{algorithmic}
\end{algorithm}

Now, let us intuitively explain Algorithm~\ref{algo:2}.
To enforce $K$-SSO in Definition~\ref{de:3.1}, we first need to know whether there exists a leaking-secret run that violates $K$-SSO in Definition~\ref{de:3.1} in $Cc(\hat{G},\tilde{G}_{obs})$.
If such a leaking-secret run exists, by Theorem~\ref{th:3.1}, in Line $8$ we compute the set $\Theta$ consisting of all states of the form $(\cdot,\emptyset)$ that are observationally reachable from $\hat{X}_{cc,0}$ within $K$ steps.
In Lines $10$ to $13$, we need to further capture whether there exists an uncontrollable leaking-secret run $(x,\tilde{q}_0)\stackrel{e}{\rightarrow}(\cdot,\emptyset)$ of observational length at most $K$ in $Cc(\hat{G},\tilde{G}_{obs})$ and its an uncontrollable predecessor run $(x_0, q_0)\stackrel{\tilde{e}}{\rightarrow}(x, q)$ in $Cc(G, Obs(G))$, where $\tilde{q}_0=q\setminus X_S$.
If such an uncontrollable leaking-secret run exists, then $G$ cannot be enforced to be $K$-SSO in Definition~\ref{de:3.1}.
Whereas if all leaking-secret runs that violate $K$-SSO in Definition~\ref{de:3.1} are controllable, in order to cut off them, our strategy is to disable the last controllable transition along each one of them\footnote{A transition $(x_1, q_1)\stackrel{{(\sigma,\tilde{\sigma}})}{\longrightarrow}(x_2, q_2)$ in $Cc(\hat{G},\tilde{G}_{obs})$ (resp., $Cc(G, Obs(G))$) is called controllable if $\sigma\in\Sigma_c$.
A controllable run in $Cc(\hat{G},\tilde{G}_{obs})$ (resp., $Cc(G, Obs(G))$) is a run that contains a controllable transition.}.
This process can be done in Lines $16$ to $18$.
Note that, the set $\Lambda$ in Line $16$ can be computed using the ``Breadth-First Search Algorithm" in~\cite{Cormen(2009)}.
It is worth noting that when we disable a controllable transition in $Cc(\hat{G},\tilde{G}_{obs})$ or $Cc(G, Obs(G))$, it may make a non-leaking-secret run to become leaking-secret. (see Example~\ref{ex:4.1} below).
To this end, in Line $23$, we repeatedly update $Ac(Cc(\hat{G},\tilde{G}_{obs}))$ and $AC(Cc(G,Obs(G)))$.

\begin{remark}\label{re:4.1}
In particular, when $K=0$, the enforcement of $0$-SSO in Definition~\ref{de:3.1} (i.e., standard CSO) in Algorithm~\ref{algo:2} can be simplified using the concurrent composition $Cc(G,Obs(G))$ based on Proposition~\ref{pro:4.1}.
$Cc(\hat{G},\tilde{G}_{obs})$ is not needed.
\end{remark}

Next, we prove the correctness of Algorithm~\ref{algo:2} for enforcing $K$-SSO in Definition~\ref{de:3.1}.

\begin{theorem}\label{th:4.1}
A subset $\mathcal{E}_c\subseteq\mathcal{T}_c^G$ returned by Algorithm~\ref{algo:2} (if exists) is a solution to Problem~\ref{problem:1} on $K$-SSO in Definition~\ref{de:3.1}.
\end{theorem}

\begin{proof}
There are two cases: 1) When $\mathcal{E}_c=\emptyset$, i.e., there exists no state of the form $(\cdot,\emptyset)$ in $Cc(\hat{G},\tilde{G}_{obs})$ that is observationally reachable from $\hat{X}_{cc,0}$ within $K$ steps.
Then by Theorem~\ref{th:3.1}, $G_{\backslash\mathcal{E}_c}=G$ is $K$-SSO.
Obviously, $\mathcal{E}_c=\emptyset$ is a solution to Problem~\ref{problem:1} on $K$-SSO in Definition~\ref{de:3.1}.
2) When $\mathcal{E}_c\neq\emptyset$, by Algorithm~\ref{algo:2}, we conclude that in the most updated $Ac(Cc(\hat{G},\tilde{G}_{obs}))$ (if exists\footnote{When Algorithm~\ref{algo:2} outputs $\mathcal{E}_c\neq\emptyset$, it is possible that all secret states in $X_S$ have been removed in the most updated $G$. In this case, $Ac(Cc(\hat{G},\tilde{G}_{obs}))$ disappears and the subsystem $G_{\backslash\mathcal{E}_c}$ is $K$-SSO in Definition~\ref{de:3.1}.}) there is no state of the form $(\cdot,\emptyset)$ that is observationally reachable from the initial-state set within $K$ steps.
Then by Theorem~\ref{th:3.1}, the subsystem $G_{\backslash\mathcal{E}_c}$ is $K$-SSO.
Taken together, Algorithm~\ref{algo:2} can correctly enforce $K$-SSO in Definition~\ref{de:3.1} for a given system $G$, a set $X_S\subseteq X$ of secret states, and a non-negative integer $K$.
\end{proof}

\begin{example}\label{ex:4.1}
Let us consider again the system $G$ depicted in Fig.~\ref{Fig1}.
According to Example~\ref{ex:3.1}, $G$ is not $2$-SSO in Definition~\ref{de:3.1}.
Now we demonstrate how Algorithm~\ref{algo:2} works for solving Problem~\ref{problem:1} on $2$-SSO.
Assume that only event $c$ is uncontrollable, i.e., $\Sigma_{uc}=\{c\}$.
In order to prevent $Cc(\hat{G},\tilde{G}_{obs})$ depicted in Fig.~\ref{Fig2}(e) from reaching state $(8,\emptyset)$ in $2$ observational steps, we need to cut off controllable transition $(7,\{1,2,3,4\})\stackrel{(b,b)}{\longrightarrow}(8,\{2\})$.
This equivalently means that controllable transition $7\stackrel{b}{\rightarrow}8$ in $G$ needs to be disabled.
Hence, we obtain the subsystem $G^1$ and its corresponding concurrent composition $Cc(\hat{G}^1,\tilde{G}^1_{obs})$ shown in~Fig.~\ref{Fig:subfig4a} and~Fig.~\ref{Fig:subfig4b}, respectively.

\begin{figure}[!ht]
\centering
\subcaptionbox{The subautomaton $G^1$\label{Fig:subfig4a}}{
    	  \begin{tikzpicture}[>=stealth',shorten >=1pt,auto,node distance=2.0 cm, scale = 0.5, transform shape,
	>=stealth,inner sep=2pt]

	\tikzstyle{emptynode}=[inner sep=0,outer sep=0]

	\node[initial, initial where = left, state] (0) {$0$};
	\node[state] (1) [right of =0] {$1$};
	\node[state] (2) [right of =1] {$2$};
	\node[state] (3) [above of =0] {$3$};
	\node[state] (4) [right of =3] {$4$};
	\node[state] (5) [right of =4] {$\red 5$};
	\node[state] (6) [below of =0] {$6$};
	\node[state] (7) [right of =6] {$\red 7$};
	\node[emptynode] (e1) [right of =5] {};
	\node[emptynode] (e2) [left of =3] {};

	\path [->]
	(0) edge node [above, sloped] {$a$} (1)
	(1) edge node [above, sloped] {$u$} (2)
	(3) edge node [above, sloped] {$u$} (4)
	(4) edge node [above, sloped] {$b$} (5)
	(6) edge node [above, sloped] {$a$} (7)
	(0) edge node {$a$} (3)
	(0) edge node {$u$} (6)
	(5) edge [loop right] node {$c$} (5)
	(2) edge [loop right] node {$b$} (2)
	;

        \end{tikzpicture}
}
\subcaptionbox{$Cc(\hat{G}^1,\tilde{G}^1_{obs})$\label{Fig:subfig4b}}{
  	  \begin{tikzpicture}[>=stealth',shorten >=1pt,auto,node distance=3.2 cm, scale = 0.5, transform shape,
	>=stealth,inner sep=2pt]

	\node[initial, initial where = above, rectangle state] (7-1234) {$({\red 7},\{1,2,3,4\})$};
	\node[initial, initial where = above, rectangle state] (5-2) [right of = 7-1234] {$({\red 5},\{2\})$};
	\node[initial, initial where = right, rectangle state] (5-phi) [below of = 5-2] {$({\red 5},\emptyset)$};

	\path [->]
	(5-2) edge node {$(c,c)$} (5-phi)
	(5-phi) edge [loop below] node {$(c,c)$} (5-phi)
	;

        \end{tikzpicture}
}
\caption{The subautomaton $G^1$ and its corresponding concurrent composition $Cc(\hat{G}^1,\tilde{G}^1_{obs})$.}
\label{Fig4}
\end{figure}
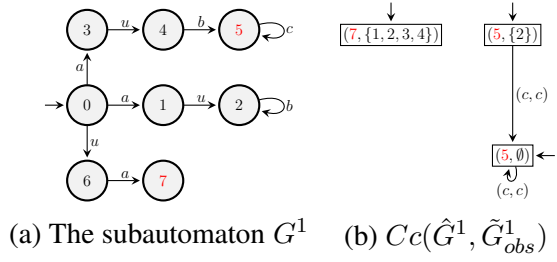

By~Fig.~\ref{Fig:subfig4b}, initial state $(5,\emptyset)$ of $Cc(\hat{G}^1,\tilde{G}^1_{obs})$ should also be removed.
Furthermore, since transition $(5,\{2\})\stackrel{(c,c)}{\longrightarrow}(5,\emptyset)$ is uncontrollable, initial state $(5,\{2\})$ should also be removed.
In order to prevent these two initial states $(5,\{2\})$ and $(5,\emptyset)$ from being reached, we construct the concurrent composition $Cc(G^1,Obs(G^1))$, as shown in~Fig.~\ref{Fig5}.
Since initial states $(5,\{2\})$ and $(5,\emptyset)$ in $Cc(\hat{G}^1,\tilde{G}^1_{obs})$ correspond to states $(5,\{2,5\})$ and $(5,\{5\})$ in $Cc(G^1,Obs(G^1))$, respectively, we cut off controllable transition $(4,\{1-4,7\})\stackrel{(b,b)}{\longrightarrow}(5,\{2,5\})$ ( i.e., $4\stackrel{b}{\rightarrow}5$ in $G_1$).
This results in the subsystem $G^2$ shown in~Fig.~\ref{Fig:subfig6a}, whose corresponding concurrent composition $Cc(\hat{G}^2,\tilde{G}^2_{obs})$ is shown in~Fig.~\ref{Fig:subfig6b}.
By Theorem~\ref{th:3.1}, $G^2=G_{\backslash\mathcal{E}_c}$ is $2$-SSO in Definition~\ref{de:3.1}, where $\mathcal{E}_c=\{4\stackrel{b}{\rightarrow}5, 7\stackrel{b}{\rightarrow}8\}$.
Thus, subset $\mathcal{E}_c=\{4\stackrel{b}{\rightarrow}5, 7\stackrel{b}{\rightarrow}8\}$ is a solution of Problem~\ref{problem:1} on $2$-SSO in Definition~\ref{de:3.1}.

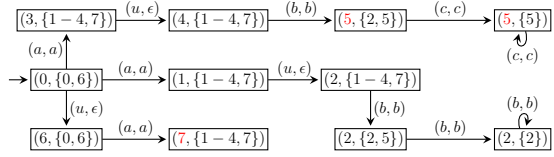
\begin{figure}[!ht]
  \centering
	\begin{tikzpicture}[>=stealth',shorten >=1pt,auto,node distance=4.0 cm, scale = 0.5, transform shape,
	>=stealth,inner sep=2pt]

	\node[initial, initial where = left, rectangle state] (0-06) {$(0,\{0,6\})$};
	\node[rectangle state] (3-1-47) [above = 1cm of 0-06] {$(3,\{1-4,7\})$};
	\node[rectangle state] (6-06) [below = 1cm of 0-06] {$(6,\{0,6\})$};
	\node[rectangle state] (1-1-47) [right of = 0-06] {$(1,\{1-4,7\})$};
	\node[rectangle state] (4-1-47) [above = 1cm of 1-1-47] {$(4,\{1-4,7\})$};
	\node[rectangle state] (7-1-47) [below = 1cm of 1-1-47] {$({\red 7},\{1-4,7\})$};
	\node[rectangle state] (5-25) [right of = 4-1-47] {$({\red 5},\{2,5\})$};
	\node[rectangle state] (5-5) [right of = 5-25] {$({\red 5},\{5\})$};
	\node[rectangle state] (2-1-47) [right of = 1-1-47] {$(2,\{1-4,7\})$};
	\node[rectangle state] (2-25) [below = 1cm of 2-1-47] {$(2,\{2,5\})$};
	\node[rectangle state] (2-2) [right of = 2-25] {$(2,\{2\})$};

	\path [->]
	(0-06) edge node {$(a,a)$} (3-1-47)
	(0-06) edge node {$(u,\epsilon)$} (6-06)
	(0-06) edge node {$(a,a)$} (1-1-47)
	(3-1-47) edge node {$(u,\epsilon)$} (4-1-47)
	(6-06) edge node {$(a,a)$} (7-1-47)
	(4-1-47) edge node {$(b,b)$} (5-25)
	(5-25) edge node {$(c,c)$} (5-5)
	(5-5) edge [loop below] node {$(c,c)$} (5-5)
	(1-1-47) edge node {$(u,\epsilon)$} (2-1-47)
	(2-1-47) edge node {$(b,b)$} (2-25)
	(2-25) edge node {$(b,b)$} (2-2)
	(2-2) edge [loop above] node {$(b,b)$} (2-2)
	;

        \end{tikzpicture}
  \caption{The concurrent composition $Cc(G^1,Obs(G^1))$ for the subautomaton $G^1$ shown in Fig~\ref{Fig:subfig4a}.}
  \label{Fig5}
\end{figure}

\begin{figure}[!ht]
\centering
\subcaptionbox{The subautomaton $G^2$
\label{Fig:subfig6a}}{
      	  \begin{tikzpicture}[>=stealth',shorten >=1pt,auto,node distance=2.0 cm, scale = 0.5, transform shape,
	>=stealth,inner sep=2pt]

	\tikzstyle{emptynode}=[inner sep=0,outer sep=0]

	\node[initial, initial where = left, state] (0) {$0$};
	\node[state] (1) [right of =0] {$1$};
	\node[state] (2) [right of =1] {$2$};
	\node[state] (3) [above of =0] {$3$};
	\node[state] (4) [right of =3] {$4$};
	\node[state] (6) [below of =0] {$6$};
	\node[state] (7) [right of =6] {$\red 7$};

	\node[emptynode] (e1) [right of =2] {};
	\node[emptynode] (e2) [left of =0] {};

	\path [->]
	(0) edge node [above, sloped] {$a$} (1)
	(1) edge node [above, sloped] {$u$} (2)
	(3) edge node [above, sloped] {$u$} (4)
	(6) edge node [above, sloped] {$a$} (7)
	(0) edge node {$a$} (3)
	(0) edge node {$u$} (6)
	(2) edge [loop right] node {$b$} (2)
	;

        \end{tikzpicture}
}
\subcaptionbox{$Cc(\hat{G}^2,\tilde{G}^2_{obs})$
\label{Fig:subfig6b}}{
    	  \begin{tikzpicture}[>=stealth',shorten >=1pt,auto,node distance=3.2 cm, scale = 0.8, transform shape,
	>=stealth,inner sep=2pt]

	\tikzstyle{emptynode}=[inner sep=0,outer sep=0]

	\node[initial, initial where = above, rectangle state] (7-1234) {$({\red 7},\{1,2,3,4\})$};
	\node[emptynode] (e1) [right of =7-1234] {};
	\node[emptynode] (e2) [left of =7-1234] {};

    \end{tikzpicture}
}
\caption{The subautomaton $G^2$ and its corresponding concurrent composition $Cc(\hat{G}^2,\tilde{G}^2_{obs})$.}
\label{Fig6}
\end{figure}
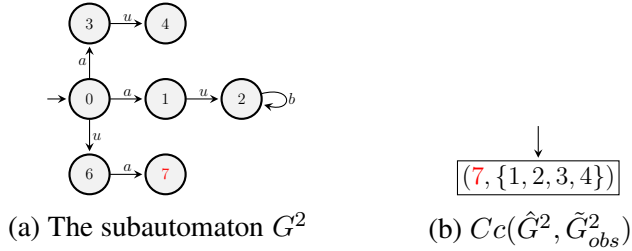
\end{example}

The following result reveals that Problem~\ref{problem:1} on $K$-SSO in Definition~\ref{de:3.1} has a solution if and only if Algorithm~\ref{algo:2} outputs ``Yes" and a subset $\mathcal{E}_c\subseteq\mathcal{T}_c^G$.

\begin{theorem}\label{th:4.2}
Given a system $G=(X,\Sigma,\delta,X_0)$, a set $\Sigma_o\subseteq\Sigma$ of observable events, a set $\Sigma_c\subseteq\Sigma$ of controllable events, a set $X_S\subseteq X$ of secret states, and a non-negative integer $K$,
if Algorithm~\ref{algo:2} outputs ``No" (i.e., no subset $\mathcal{E}_c\subseteq\mathcal{T}_c^G$ is returned by Algorithm~\ref{algo:2}), then Problem~\ref{problem:1} on $K$-SSO in Definition~\ref{de:3.1} has no solution.
\end{theorem}

\begin{proof}
Assume that Algorithm~\ref{algo:2} outputs ``No".
Then there are two cases.

\textbf{1)} There exists at least one uncontrollable leaking-secret run $(x,\tilde{q}_0)\stackrel{e}{\rightarrow}(\cdot,\emptyset)$ of observational length at most $K$ in $Cc(\hat{G},\tilde{G}_{obs})$ and at least one uncontrollable predecessor run $(x_0, q_0)\stackrel{\tilde{e}}{\rightarrow}(x, q)$ in $Cc(G, Obs(G))$, where $(x,\tilde{q}_0)\in\hat{X}_{cc,0}$, $(x_0,q_0)$ is an initial state of $Cc(G,Obs(G))$, and $\tilde{q}_0=q\setminus X_S$.
By the constructions of both $Cc(\hat{G},\tilde{G}_{obs})$ and $Cc(G, Obs(G))$, we conclude that when $G$ generates the run $x_0\stackrel{\tilde{e}(L)}{\rightarrow}x\stackrel{e(L)}{\rightarrow}\cdot$, an intruder determines for sure that $G$ is at some secret states after $\tilde{e}(L)$ has just been generated within additional $K$ observable steps given that the observations are $P(\tilde{e}(L)e(L))$.
Since the leaking-secret run $x_0\stackrel{\tilde{e}(L)}{\rightarrow}x\stackrel{e(L)}{\rightarrow}\cdot$ is uncontrollable, it cannot be prevented from being generated.
Hence, Problem~\ref{problem:1} on $K$-SSO in Definition~\ref{de:3.1} has no solution.

\textbf{2)} Assume initially that there exists at least one controllable leaking-secret run, and there exists no uncontrollable leaking-secret run, but after Algorithm~\ref{algo:2} runs for a while, an uncontrollable leaking-secret run appears.
In this case, Algorithm~\ref{algo:2} outputs ``No".
Denote the secret runs corresponding to this process by $r^{(i)}:=x_0^{(i)}\xrightarrow[]{\tilde e^{(i)}}x_1^{(i)}\xrightarrow[]{e^{(i)}}x_2^{(i)}$, where $x_0^{(i)}\in X_0$, $x_1^{(i)}\in X_S$, $|P(e^{(i)})|\le K$, $i=1,2,\dots,n$,
$\tilde e^{(j)}e^{(j)}$ contains at least one controllable event, $j=1,2,\dots,n-1$, $\tilde e^{(n)}e^{(n)}$ contains no controllable event.
Also assume initially only $r^{(1)}$ is a leaking-secret run among $r^{(1)},r^{(2)},\dots,r^{(n)}$.
Assume after the {\bf for}-loop since Line~\ref{line1} starts, after the $j$-th execution of the {\bf for}-loop, $j=1,2,\dots,n-1$,
the last controllable transition of run $r^{(j)}$ is disabled, and then run $r^{(j+1)}$ becomes leaking-secret.
Then finally, $r^{(n)}$ becomes an uncontrollable leaking-secret run, and Algorithm~\ref{algo:2} returns ``No".

Next, we show that if during each execution of the {\bf for}-loop, Algorithm~\ref{algo:2} does not necessarily disable the ``last'' controllable transitions, but arbitrarily choose controllable transitions to disable, Algorithm~\ref{algo:2} still will return ``No".
Now we make the setting of the {\bf for}-loop to be more flexible as follows: in the first execution of the {\bf for}-loop, we do not only disable the last controllable transition of $r^{(1)}$, but choose an arbitrary number of controllable transitions (may including the last controllable transition of $r^{(1)}$) to disable to cut off all runs whose label sequences are the same as the label sequence of $r^{(1)}$, then we will have all runs $r^{(k)}$ are cut off with $1\le k< l_1$ for some $1< l_1\le n$, and run $r^{(l_1)}$ becomes a leaking-secret run.
Then we do the same thing to $r^{(l_1)}$, and will have $r^{(l_2)}$ becomes a leaking-secret run, where $l_1<l_2\le n$.
After repeating this for no more than $n-1$ times, we will have all runs $r^{(i)}$ ($1\le i<n$) are cut off, and $r^{(n)}$ becomes an uncontrollable leaking-secret run.
Then we cannot make the currently updated $G$ become $K$-SSO in Definition~\ref{de:3.1} by choosing an arbitrary number of controllable transitions to disable.
Therefore, system $G$ cannot be enforced to be $K$-SSO in Definition~\ref{de:3.1}.
\end{proof}

\begin{example}\label{ex:4.1a}
Consider the system $G$ depicted in Fig.~\ref{Fig7} in which $\Sigma_o=\{a,b,c\}$, $\Sigma_c=\{u,v\}$, $X_0=\{0\}$, and $X_S=\{4,13,16\}$.
The corresponding concurrent composition $Cc(\hat{G},\tilde{G}_{obs})$ is shown in Fig.~\ref{Fig8}.
By Theorem~\ref{th:3.1}, $G$ is $1$-SSO in Definition~\ref{de:3.1}, but not $2$-SSO.

\begin{figure}[!ht]
  \centering
	\begin{tikzpicture}[>=stealth',shorten >=1pt,auto,node distance=2.0 cm, scale = 0.5, transform shape,
	>=stealth,inner sep=2pt]

	\node[initial, initial where = left, state] (0) {$0$};
	\node[state] (1) [right of =0] {$1$};
	\node[state] (2) [right of =1] {$2$};
	\node[state] (3) [right of =2] {$3$};
	\node[state] (4) [right of =3] {$\red 4$};
	\node[state] (5) [right of =4] {$5$};
	\node[state] (6) [below of =1] {$6$};
	\node[state] (7) [below of =2] {$7$};
	\node[state] (14) [below of =6] {$14$};
	\node[state] (15) [below of =7] {$15$};
	\node[state] (16) [right of =15] {$\red 16$};
	\node[state] (8) [above of =0] {$8$};
	\node[state] (9) [right of =8] {$9$};
	\node[state] (10) [right of =9] {$10$};
	\node[state] (11) [right of =10] {$11$};
	\node[state] (12) [right of =11] {$12$};
	\node[state] (13) [right of =12] {$\red13$};

	\path [->]
	(0) edge node [above, sloped] {$a$} (1)
	(1) edge node [above, sloped] {$u$} (2)
	(2) edge node [above, sloped] {$c$} (3)
	(3) edge node [above, sloped] {$v$} (4)
	(4) edge node {$a$} (5)
	(5) edge [loop right] node {$a,b$} (5)
	(0) edge [bend right] node {$u$} (6)
	(6) edge node {$a$} (7)
	(7) edge [bend right] node {$c$} (3)
	(0) edge [bend right] node {$a$} (14)
	(14) edge node {$c$} (15)
	(15) edge node {$a$} (16)
	(0) edge node {$v$} (8)
	(8) edge node {$a$} (9)
	(9) edge node {$u$} (10)
	(10) edge node {$c$} (11)
	(11) edge node {$a$} (12)
	(12) edge node {$\tau$} (13)
	(13) edge [loop right] node {$b$} (13)
	;

        \end{tikzpicture}
  \caption{The automaton $G$ considered in Example~\ref{ex:4.1a}.}
  \label{Fig7}
\end{figure}
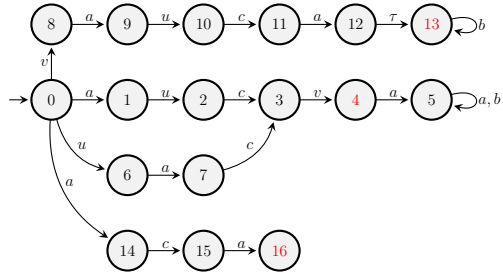

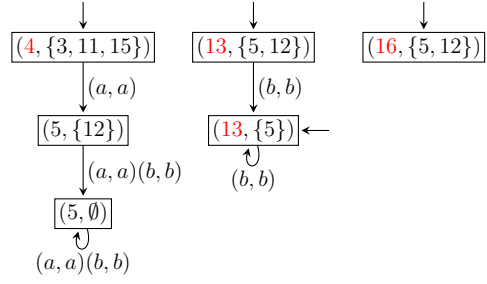
\begin{figure}[!ht]
  \centering
  	\begin{tikzpicture}[>=stealth',shorten >=1pt,auto,node distance=3.2 cm, scale = 0.7, transform shape,
	>=stealth,inner sep=2pt]

	\node[initial, initial where = above, rectangle state] (4-31115) {$({\red 4},\{3,11,15\})$};
	\node[initial, initial where = above, rectangle state] (13-512) [right of = 4-31115] {$({\red 13},\{5,12\})$};
	\node[initial, initial where = above, rectangle state] (16-512) [right of = 13-512] {$({\red 16},\{5,12\})$};
	\node[rectangle state] (5-12) [below = 1cm of 4-31115] {$(5,\{12\})$};
	\node[initial, initial where = right, rectangle state] (13-5) [below = 1cm of 13-512] {$({\red 13},\{5\})$};
	\node[rectangle state] (5-phi) [below = 1cm of 5-12] {$(5,\emptyset)$};

	\path [->]
	(4-31115) edge node {$(a,a)$} (5-12)
	(5-12) edge node {$(a,a)(b,b)$} (5-phi)
	(5-phi) edge [loop below] node {$(a,a)(b,b)$} (5-phi)
	(13-512) edge node {$(b,b)$} (13-5)
	(13-5) edge [loop below] node {$(b,b)$} (13-5)
	;

  \end{tikzpicture}
  \caption{The concurrent composition $Cc(\hat{G},\tilde{G}_{obs})$ for the automaton $G$ shown in Fig.~\ref{Fig7}.}
  \label{Fig8}
\end{figure}

Next, we want to know whether the enforcement of $2$-SSO of $G$ can be done using Algorithm~\ref{algo:2}.
Specifically, in order to prevent state $(5,\emptyset)$ in $Cc(\hat{G},\tilde{G}_{obs})$ from being reached within $2$ observational steps,
states $(4,\{3,11,15\})$ and $(5,\{12\})$ should be removed since events $a$ and $b$ are uncontrollable.
To this end, we construct the concurrent composition $Cc(G,Obs(G))$.
For the sake of simplicity, it is enough to depict part of $Cc(G,Obs(G))$ in Fig.~\ref{Fig9}.
Since state $(4,\{3,11,15\})$ in $Cc(\hat{G},\tilde{G}_{obs})$ corresponds to state $(4,\{3,4,11,15\})$ in $Cc(G,Obs(G))$, we cut off controllable transition
$(3,\{3,4,11,15\})\stackrel{(v,\epsilon)}{\longrightarrow}(4,\{3,4,11,15\})$ ( i.e., $3\stackrel{v}{\rightarrow}4$ in $G$).
This results in the subsystem $G^1$ shown in Fig.~\ref{Fig:subfig10a}, whose corresponding concurrent composition $Cc(\hat{G}^1,\tilde{G}^1_{obs})$ is shown in Fig.~\ref{Fig:subfig10b}.
By Theorem~\ref{th:3.1}, states $(13,12)$ and $(13,\emptyset)$ in $Cc(\hat{G}^1,\tilde{G}^1_{obs})$ should be forbidden to be reached.
Therefore, we construct part of the concurrent composition $Cc(G^1,Obs(G^1))$, as shown in Fig.~\ref{Fig11}.
In order to prevent states $(13,\{12,13,16\})$ and $(13,\{13\})$ in $Cc(G^1,Obs(G^1))$ from being reached, we cut off controllable transition $(9,\{1,2,7,9,10,14\})\stackrel{(u,\epsilon)}{\longrightarrow}(10,\{1,2,7,9,10,14\})$ ( i.e., $9\stackrel{u}{\rightarrow}10$ in $G^1$).
We obtain the subsystem $G^2$ shown in Fig.~\ref{Fig12}, whose corresponding concurrent composition $Cc(\hat{G}^2,\tilde{G}^2_{obs})$ is shown in Fig.~\ref{Fig12}.
In $Cc(\hat{G}^2,\tilde{G}^2_{obs})$, initial state $(16,\emptyset)$ should be forbidden to be reached.
Toward that end, we construct part of the concurrent composition $Cc(G^2,Obs(G^2))$ in Fig.~\ref{Fig13} in which state $(16,\{16\})$ corresponds to state $(16,\emptyset)$ in $Cc(\hat{G}^2,\tilde{G}^2_{obs})$.
Since the only run ended with state $(16,\{16\})$ is uncontrollable, Algorithm~\ref{algo:2} outputs ``No".
Therefore, $G$ cannot be enforced to be $2$-SSO.
Note that, it is not difficult to conclude that $G$ cannot be enforced to be $2$-SSO no matter how to cut off its controllable transitions.

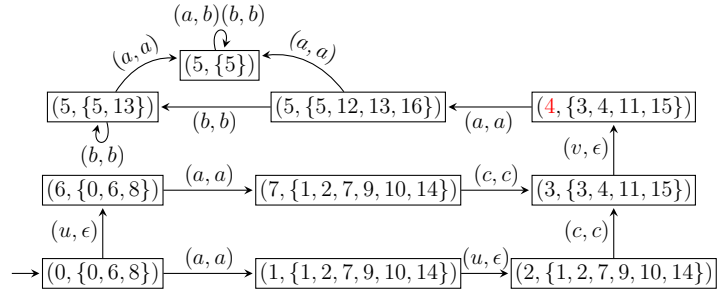
\begin{figure}[!ht]
  \centering
  	\begin{tikzpicture}[>=stealth',shorten >=1pt,auto,node distance=4.8 cm, scale = 0.7, transform shape,
	>=stealth,inner sep=2pt]

	\node[rectangle state] (6-068) {$(6,\{0,6,8\})$};
	\node[rectangle state] (5-513) [above = 1cm of 6-068] {$(5,\{5,13\})$};
	\node[initial, initial where = left, rectangle state] (0-068) [below = 1cm of 6-068] {$(0,\{0,6,8\})$};
	\node[rectangle state] (7-12791014) [right of = 6-068] {$(7,\{1,2,7,9,10,14\})$};
	\node[rectangle state] (5-5121316) [above = 1cm of 7-12791014] {$(5,\{5,12,13,16\})$};
	\node[rectangle state] (1-12791014) [below = 1cm of 7-12791014] {$(1,\{1,2,7,9,10,14\})$};
	\node[rectangle state] (4-341115) [right of = 5-5121316] {$({\red 4},\{3,4,11,15\})$};
	\node[rectangle state] (3-341115) [right of = 7-12791014] {$(3,\{3,4,11,15\})$};
	\node[rectangle state] (2-12791014) [below = 1cm of 3-341115] {$(2,\{1,2,7,9,10,14\})$};
	\node[rectangle state] (5-5) [above left = 0.3cm of 5-5121316] {$(5,\{5\})$};

	\path [->]
	(0-068) edge node {$(u,\epsilon)$} (6-068)
	(6-068) edge node {$(a,a)$} (7-12791014)
	(5-5121316) edge node {$(b,b)$} (5-513)
	(0-068) edge node {$(a,a)$} (1-12791014)
	(4-341115) edge node {$(a,a)$} (5-5121316)
	(7-12791014) edge node {$(c,c)$} (3-341115)
	(2-12791014) edge node {$(c,c)$} (3-341115)
	(1-12791014) edge node {$(u,\epsilon)$} (2-12791014)
	(3-341115) edge node {$(v,\epsilon)$} (4-341115)
	(5-513) edge [loop below] node {$(b,b)$} (5-513)
	(5-5) edge [loop above] node {$(a,b)(b,b)$} (5-5)
	(5-513) edge [bend left] node [sloped, above] {$(a,a)$} (5-5)
	(5-5121316) edge [bend right] node [sloped, above] {$(a,a)$} (5-5)
	;

  \end{tikzpicture}
  \caption{Part of the concurrent composition $Cc(G,Obs(G))$ for the automaton $G$ shown in Fig.~\ref{Fig7}.}
  \label{Fig9}
\end{figure}

\begin{figure}[!ht]
\centering
\subcaptionbox{The subautomaton $G^1$\label{Fig:subfig10a}}{
  	\begin{tikzpicture}[>=stealth',shorten >=1pt,auto,node distance=2.0 cm, scale = 0.5, transform shape,
	>=stealth,inner sep=2pt]

	\node[initial, initial where = left, state] (0) {$0$};
	\node[state] (1) [right of =0] {$1$};
	\node[state] (2) [right of =1] {$2$};
	\node[state] (3) [right of =2] {$3$};
	\node[state] (6) [below of =1] {$6$};
	\node[state] (7) [below of =2] {$7$};
	\node[state] (14) [below of =6] {$14$};
	\node[state] (15) [below of =7] {$15$};
	\node[state] (16) [right of =15] {$\red 16$};
	\node[state] (8) [above of =0] {$8$};
	\node[state] (9) [right of =8] {$9$};
	\node[state] (10) [right of =9] {$10$};
	\node[state] (11) [right of =10] {$11$};
	\node[state] (12) [right of =11] {$12$};
	\node[state] (13) [below of =12] {$\red13$};

	\path [->]
	(0) edge node [above, sloped] {$a$} (1)
	(1) edge node [above, sloped] {$u$} (2)
	(2) edge node [above, sloped] {$c$} (3)
	(0) edge [bend right] node {$u$} (6)
	(6) edge node {$a$} (7)
	(7) edge [bend right] node {$c$} (3)
	(0) edge [bend right] node {$a$} (14)
	(14) edge node {$c$} (15)
	(15) edge node {$a$} (16)
	(0) edge node {$v$} (8)
	(8) edge node {$a$} (9)
	(9) edge node {$u$} (10)
	(10) edge node {$c$} (11)
	(11) edge node {$a$} (12)
	(12) edge node {$\tau$} (13)
	(13) edge [loop right] node {$b$} (13)
	;

    \end{tikzpicture}
}
\subcaptionbox{$Cc(\hat{G}^1,\tilde{G}^1_{obs})$\label{Fig:subfig10b}}{
   \begin{tikzpicture}[>=stealth',shorten >=1pt,auto,node distance=3.2 cm, scale = 0.7, transform shape,
	>=stealth,inner sep=2pt]

	\node[initial, initial where = above, rectangle state] (13-12) [right of = 4-31115] {$({\red 13},\{12\})$};
	\node[initial, initial where = above, rectangle state] (16-12) [right of = 13-12] {$({\red 16},\{12\})$};
	\node[initial, initial where = right, rectangle state] (13-) [below = 1cm of 13-12] {$({\red 13},\emptyset)$};

	\path [->]
	(13-12) edge node {$(b,b)$} (13-)
	(13-) edge [loop below] node {$(b,b)$} (13-)
	;

  \end{tikzpicture}
}
\caption{The subautomaton $G^1$ and its corresponding concurrent composition $Cc(\hat{G}^1,\tilde{G}^1_{obs})$.}
\label{Fig10}
\end{figure}
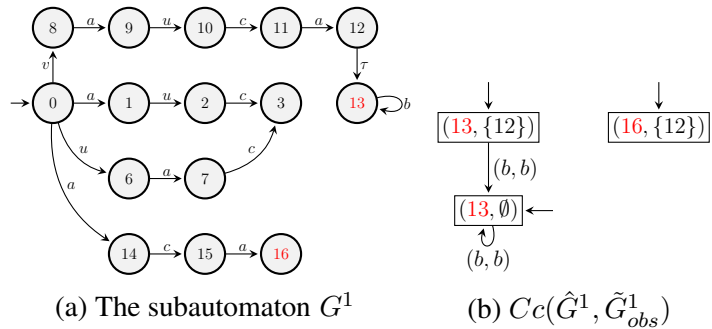

\begin{figure}[!ht]
  \centering 
  \begin{tikzpicture}[>=stealth',shorten >=1pt,auto,node distance=4.3 cm, scale = 0.7, transform shape,
	>=stealth,inner sep=2pt]

	\node[initial, initial where = left, rectangle state] (0-068) {$(0,\{0,6,8\})$};
	\node[rectangle state] (8-068) [right of = 0-068] {$(8,\{0,6,8\})$};
	\node[rectangle state] (9-12791014) [right of = 8-068] {$(9,\{1,2,7,9,10,14\})$};
	\node[rectangle state] (10-12791014) [below = 1cm of 9-12791014] {$(10,\{1,2,7,9,10,14\})$};
	\node[rectangle state] (11-31115) [left of = 10-12791014] {$(11,\{3,11,15\})$};
	\node[rectangle state] (13-13) [below = 1cm of 10-12791014] {$({\red 13},\{13\})$};
	\node[rectangle state] (13-121316) [left of = 13-13] {$({\red 13},\{12,13,16\})$};
	\node[rectangle state] (12-121316) [left of = 13-121316] {$(12,\{12,13,16\})$};

	\path [->]
	(12-121316) edge node {$(\tau,\epsilon)$} (13-121316)
	(0-068) edge node {$(v,\epsilon)$} (8-068)
	(13-121316) edge node {$(b,b)$} (13-13)
	(10-12791014) edge node {$(c,c)$} (11-31115)
	(9-12791014) edge node {$(u,\epsilon)$} (10-12791014)
	(8-068) edge node {$(a,a)$} (9-12791014)
	(11-31115) edge node [sloped, above] {$(a,a)$} (12-121316)
	(13-13) edge [loop below] node {$(b,b)$} (13-13)
	;
  \end{tikzpicture}

  \caption{Part of the concurrent composition $Cc(G_1,Obs(G_1))$ for the subautomaton $G_1$ shown in Fig.~\ref{Fig10}.}
  \label{Fig11}
\end{figure}
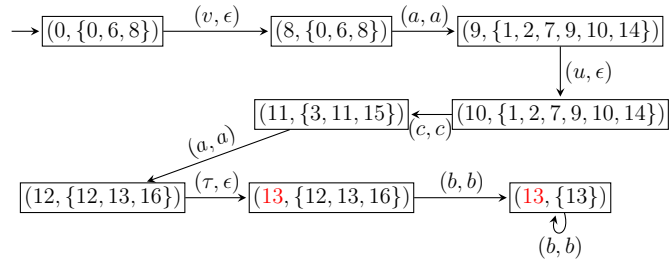

\begin{figure}[!ht]
\centering
\subcaptionbox{The subautomaton $G^2$\label{Fig:subfig12a}}{
    	\begin{tikzpicture}[>=stealth',shorten >=1pt,auto,node distance=2.0 cm, scale = 0.5, transform shape,
	>=stealth,inner sep=2pt]

	\node[initial, initial where = left, state] (0) {$0$};
	\node[state] (1) [right of =0] {$1$};
	\node[state] (2) [right of =1] {$2$};
	\node[state] (3) [right of =2] {$3$};
	\node[state] (6) [below of =1] {$6$};
	\node[state] (7) [below of =2] {$7$};
	\node[state] (14) [below of =6] {$14$};
	\node[state] (15) [below of =7] {$15$};
	\node[state] (16) [right of =15] {$\red 16$};
	\node[state] (8) [above of =1] {$8$};
	\node[state] (9) [right of =8] {$9$};

	\path [->]
	(0) edge node [above, sloped] {$a$} (1)
	(1) edge node [above, sloped] {$u$} (2)
	(2) edge node [above, sloped] {$c$} (3)
	(0) edge [bend right] node {$u$} (6)
	(6) edge node {$a$} (7)
	(7) edge [bend right] node {$c$} (3)
	(0) edge [bend right] node {$a$} (14)
	(14) edge node {$c$} (15)
	(15) edge node {$a$} (16)
	(0) edge node {$v$} (8)
	(8) edge node {$a$} (9)
	;

    \end{tikzpicture}
}\hspace{1cm}
\subcaptionbox{$Cc(\hat{G}^2,\tilde{G}^2_{obs})$\label{Fig:subfig12b}}{
     \begin{tikzpicture}[>=stealth',shorten >=1pt,auto,node distance=3.2 cm, scale = 0.7, transform shape,
	>=stealth,inner sep=2pt]
	\tikzstyle{emptynode}=[inner sep=0,outer sep=0]

	\node[initial, initial where = above, rectangle state] (16-) {$({\red 16},\emptyset)$};
	\node[emptynode] (1) [left of = 16-] {};
	\node[emptynode] (1) [right of = 16-] {};

  \end{tikzpicture}
}
\caption{The subautomaton $G^2$ and its corresponding concurrent composition $Cc(\hat{G}^2,\tilde{G}^2_{obs})$.}
\label{Fig12}
\end{figure}
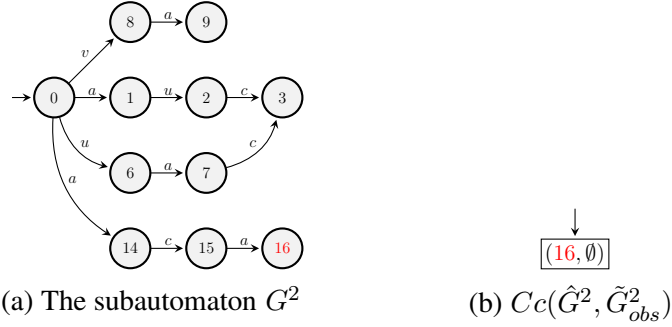

\begin{figure}[!ht]
  \centering
    \begin{tikzpicture}[>=stealth',shorten >=1pt,auto,node distance=4.3 cm, scale = 0.7, transform shape,
	>=stealth,inner sep=2pt]

	\node[initial, initial where = left, rectangle state] (0-068) {$(0,\{0,6,8\})$};
	\node[rectangle state] (14-12791014) [right of = 0-068] {$(14,\{1,2,7,9,10,14\})$};
	\node[rectangle state] (15-315) [right of = 14-12791014] {$(15,\{3,15\})$};
	\node[rectangle state] (16-16) [right of = 15-315] {$(16,\{16\})$};

	\path [->]
	(0-068) edge node {$(a,a)$} (14-12791014)
	(15-315) edge node {$(a,a)$} (16-16)
	(14-12791014) edge node {$(c,c)$} (15-315)
	;
  \end{tikzpicture}

  \caption{Part of the concurrent composition $Cc(G_2,Obs(G_2))$ for the subautomaton $G_2$ shown in Fig.~\ref{Fig12}.}
  \label{Fig13}
\end{figure}
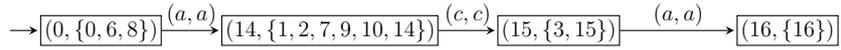
\end{example}

\begin{remark}\label{re:4.2}
We analyze the time complexity of using Algorithm~\ref{algo:2} to do enforcement for $K$-SSO in Definition~\ref{de:3.1}.
Firstly, the time consumption for computing $Cc(\hat{G},\tilde{G}_{obs})$ is $\mathcal{O}((|\Sigma_o||\Sigma_{uo}|+|\Sigma|)|X|^{2}2^{|X|})$ (See Remark~\ref{re:3.4} for details).
Performing the process of Lines $3$ to $10$ can be done in the size of $Cc(\hat{G},\tilde{G}_{obs})$, i.e., $\mathcal{O}(|\Sigma||X|^22^{|X|}+|X|2^{|X|})$.
Secondly, computing $Cc(G,Obs(G))$ has also time complexity $\mathcal{O}((|\Sigma_o||\Sigma_{uo}|+|\Sigma|)|X|^{2}2^{|X|})$.
Further, the time cost for performing the process of Lines $11$ to $16$ is $\mathcal{O}(|\Sigma||X|^22^{|X|}+|X|2^{|X|})$ using the ``Breadth-First Search Algorithm" in~\cite{Cormen(2009)}.
Thirdly, in order to cut off all runs generated by $Cc(\hat{G},\tilde{G}_{obs})$ or $Cc(G, Obs(G))$ that violate $K$-SSO in Definition~\ref{de:3.1}, in Lines $17$ to $18$ we need to disable each controllable transition in $\Lambda$, which has time complexity $\mathcal{O}(|X|^2|\Sigma_c|)$.
In Line $23$, updating $Ac(Cc(\hat{G},\tilde{G}_{obs}))$ and $Ac(Cc(G,Obs(G)))$ takes time $\mathcal{O}((|\Sigma_o||\Sigma_{uo}|+|\Sigma|)|X|^{2}2^{|X|})$.
Taken together, the overall time complexity for enforcing $K$-SSO in Definition~\ref{de:3.1} using Algorithm~\ref{algo:2} is $\mathcal{O}((|\Sigma_o||\Sigma_{uo}|+|\Sigma|)|\Sigma_c||X|^42^{|X|})$.
\end{remark}

\begin{remark}\label{re:4.3}
Analogously, two notions of $K$-SSO in Definitions~\ref{de:3.1a} and~\ref{de:3.1b} can also be enforced under the proposed enforcement framework using their verification structures in~\cite{Zhang(2023a)} and~\cite{Zhang(2023b)}.
In particular, the authors in~\cite{Ma(2021)} considered the enforcement of $K$-SSO in Definition~\ref{de:3.1a} under the supervisory control framework using the so-called $K$-step recognizer ($K$-SR).
They provided an approach to synthesize a supremal permissive supervisor for enforcing $K$-SSO in Definition~\ref{de:3.1a} under the assumption of $\Sigma_c\subseteq\Sigma_o$, which has time complexity $\mathcal{O}(|\Sigma_o|2^{(K+2)|X|})$.
Whereas our proposed enforcement approach for $K$-SSO in Definition~\ref{de:3.1a} works in a remarkably different way and does not require any assumption.
Furthermore, our approach has lower time complexity than that of~\cite{Ma(2021)}.
\end{remark}

\subsection{Enforcement of SCSO, SISO, and Inf-SSO}\label{subsec4.3}

In this subsection, we focus on the enforcement problems of SCSO, SISO, and Inf-SSO.
We first recall the formal definitions and verification approaches of these three strong SBO.
Details can be found in our previous work~\cite{Han(2023)}.

\begin{definition}[\cite{Han(2023)}]\label{de:4.2}
Given a system $G=(X,\Sigma,\delta,X_0)$, a projection map $P$ w.r.t. a set $\Sigma_o$ of observable events, and a set $X_{S}\subseteq X$ of secret states, $G$ is said to be
\begin{enumerate}
  \item \emph{strongly current-state opaque} (SCSO) w.r.t. $\Sigma_o$ and $X_{S}$, if for all $x_0\in X_0$ and all $s\in\mathcal L(G,x_0)$ with $\delta(x_0,s)\cap X_S\neq\emptyset$,
there exists a non-secret run $x^\prime_0\stackrel{t}{\rightarrow}x$ such that $P(t)=P(s)$, where $x^\prime_0\in X^{NS}_0$ and $x\in X$.
  \item \emph{strongly initial-state opaque} (SISO) w.r.t. $\Sigma_o$ and $X_{S}$, if for all $x_0\in X^S_0$ and all $s\in\mathcal{L}(G,x_0)$,
there exists a non-secret run $x^\prime_0\stackrel{t}{\rightarrow}$ such that $P(t)=P(s)$, where $x^\prime_0\in X^{NS}_0$.
  \item \emph{strongly infinite-step opaque} (Inf-SSO) w.r.t. $\Sigma_o$ and $X_{S}$, if for all $x_0\in X_0$ and all $s=s_1s_2\in\mathcal{L}(G,x_0)$ with $\delta(x_0,s_1)\cap X_S\neq\emptyset$,
there exists a non-secret run $x^\prime_0\stackrel{t}{\rightarrow}$ such that $P(t)=P(s)$, where $x^\prime_0\in X^{NS}_0$.
\end{enumerate}
\end{definition}

In our previous work~\cite{Han(2023)}, we demonstrated that: 1) SCSO and SISO are incomparable (i.e., SCSO does not imply SISO, and vice versa);
2) SCSO (resp., SISO) and $K$-SSO in Definition~\ref{de:3.1a} are also incomparable when $K\geq 1$ (resp., $K\geq 0$); and
3) Inf-SSO implies SCSO (resp., SISO), but the converse is not true.
Furthermore, we proposed a novel concurrent-composition structure $Cc(G,Obs(G_{dss}))=(X_{cc},\Sigma_{cc},\delta_{cc},X_{cc,0})$ to simultaneously verify SCSO, SISO, and Inf-SSO (see Lemma~\ref{lemma:4.1} below), where $G_{dss}$ is the \emph{non-secret subautomaton} of $G$, which can be obtained from $G$ by deleting all states in $X_S$ and computing the accessible part of the remainder.

\begin{lemma}[\cite{Han(2023)}]\label{lemma:4.1}
Given a system $G=(X,\Sigma,\delta,X_0)$, a projection map $P$ w.r.t. the set $\Sigma_o$ of observable events, and a set $X_{S}\subseteq X$ of secret states,
let $Cc(G,Obs(G_{dss}))$ be the corresponding concurrent composition.
Then $G$ is
\begin{enumerate}
  \item SCSO w.r.t. $\Sigma_o$ and $X_S$ if and only if there exists no state of the form $(x,\emptyset)$ in $Cc(G,Obs(G_{dss}))$, where $x\in X_S$.
  \item SISO w.r.t. $\Sigma_o$ and $X_S$ if and only if there exists no state of the form $(\cdot,\emptyset)$ in $Cc(G,Obs(G_{dss}))$ that is reachable from $X_{cc,0}^{S}=\{(x_0,q)\in X_{cc,0}: x_0\in X^S_0\}$.
  \item Inf-SSO w.r.t. $\Sigma_o$ and $X_S$ if and only if there exists no state of the form $(\cdot,\emptyset)$ in $Cc(G,Obs(G_{dss}))$.
\end{enumerate}
\end{lemma}

Next, we discuss how to use our previously-proposed structure $Cc(G,Obs(G_{dss}))$ given in~\cite{Han(2023)} to design the algorithms of enforcing SCSO, SISO, and Inf-SSO under the proposed framework.
Note that, the procedures of enforcing SISO and Inf-SSO are very similar to that of SCSO.
For brevity, we only formally design Algorithm~\ref{algo:3} to enforce SCSO.

\begin{algorithm}
\caption{Enforcement of SCSO}\label{algo:3}
\begin{algorithmic}[1]
\REQUIRE System $G=(X,\Sigma,\delta,X_0)$, set $\Sigma_o\subseteq\Sigma$ of observable events, set $\Sigma_c\subseteq\Sigma$ of controllable events, and set $X_{S}\subseteq X$ of secret states
\ENSURE ``Yes" if $G$ can be enforced to be SCSO, ``No" otherwise; in case of ``Yes", a subset $\mathcal{E}_c\subseteq\mathcal{T}_c^G$ of controllable transitions such that subsystem $G_{\backslash\mathcal{E}_c}$ is SCSO

\STATE Compute $Cc(G,Obs(G_{dss}))$
\STATE Initialize $\mathcal{E}_c=\emptyset$
\STATE Collect all states of the form $(x,\emptyset)$ with $x\in X_S$ in $Cc(G,Obs(G_{dss}))$ as a set $\Delta$
\IF{$\Delta$ is empty}
\STATE $G$ is SCSO
\RETURN ``Yes" and ``$\mathcal{E}_c=\emptyset$"
\STOP
\ELSE
\WHILE{$\Delta$ is not empty}
\FOR{each state of form $(x,\emptyset)$ in $\Delta$}
\IF{in $Cc(G,Obs(G_{dss}))$ there exists an uncontrollable leaking-secret run $(x_0, q)\stackrel{e}{\rightarrow}(x,\emptyset)$, where $(x_0,q)\in X_{cc,0}$ and $x\in X_S$}
\RETURN ``No"
\STOP
\ELSE
\STATE Compute the set $\Omega$ consisting of all last controllable transitions of all controllable runs in $Cc(G,Obs(G_{dss}))$ from $X_{cc,0}$ to $\Delta$
\FOR{each controllable transition $(x_1,q_1)\stackrel{{(\sigma,\tilde{\sigma}})}{\longrightarrow}(x_2, q_2)$ in $\Omega$, where $\sigma\in\Sigma_c$}
\STATE Cut off $(x_1,q_1)\stackrel{{(\sigma,\tilde{\sigma}})}{\longrightarrow}(x_2, q_2)$ by disabling its left component $x_1\stackrel{\sigma}{\rightarrow}x_2$ in $G$, add $x_1\stackrel{\sigma}{\rightarrow}x_2$ to $\mathcal{E}_c$
\ENDFOR
\ENDIF
\ENDFOR
\STATE Update $Ac(Cc(G,Obs(G_{dss})))$ and $\Delta$
\ENDWHILE
\ENDIF
\end{algorithmic}
\end{algorithm}

We prove that Algorithm~\ref{algo:3} is correct for enforcing SCSO.

\begin{theorem}\label{th:4.3}
A subset $\mathcal{E}_c\subseteq\mathcal{T}_c^G$ returned by Algorithm~\ref{algo:3} (if exists) is a solution to Problem~\ref{problem:1} on SCSO.
\end{theorem}

\begin{proof}
1) When $\mathcal{E}_c=\emptyset$, by Algorithm~\ref{algo:3} there exists no state of the form $(x,\emptyset)$ in $Cc(G,Obs(G_{dss}))$, where $x\in X_S$.
Then by Lemma~\ref{lemma:4.1}, subsystem $G_{\backslash\mathcal{E}_c}=G$ is SCSO.
Obviously, $\mathcal{E}_c=\emptyset$ is a solution to Problem~\ref{problem:1} on SCSO.
2) When $\mathcal{E}_c\neq\emptyset$, by Algorithm~\ref{algo:3} we conclude that in the most updated $Ac(Cc(G,Obs(G_{dss})))$ there exists no state of the form $(x,\emptyset)$, where $x\in X_S$.
Then by Lemma~\ref{lemma:4.1}, subsystem $G_{\backslash\mathcal{E}_c}$ still is SCSO.
Taken together, Algorithm~\ref{algo:3} can correctly enforce SCSO for a given system $G$ and a set $X_S\subseteq X$ of secret states.
\end{proof}

The following result reveals that if Problem~\ref{problem:1} on SCSO is solvable, then Algorithm~\ref{algo:3} outputs ``Yes" and a subset $\mathcal{E}_c\subseteq\mathcal{T}_c^G$.

\begin{theorem}\label{th:4.4}
Given a system $G=(X,\Sigma,\delta,X_0)$, a set $\Sigma_o\subseteq\Sigma$ of observable events, a set $\Sigma_c\subseteq\Sigma$ of controllable events, and a set $X_S\subseteq X$ of secret states,
if Algorithm~\ref{algo:3} outputs ``No" (i.e., no subset $\mathcal{E}_c\subseteq\mathcal{T}_c^G$ is returned by Algorithm~\ref{algo:3}), then Problem~\ref{problem:1} on SCSO has no solution.
\end{theorem}

\begin{proof}
This proof is similar to that of Theorem~\ref{th:4.2}.
The only difference is that we need to replace $Cc(\hat{G},\tilde{G}_{obs})$ and $Cc(G, Obs(G))$ with $Cc(G,Obs(G_{dss}))$.
We here omit it.
\end{proof}

\begin{example}\label{ex:4.2}
Consider the system $G$ depicted in Fig.~\ref{Fig14} in which $\Sigma_o=\{a,b\}$, $\Sigma_c=\{a,u,v\}$, $X_0=\{0,1\}$, and $X_S=\{1,5,9\}$.
The corresponding concurrent composition $Cc(G,Obs(G_{dss}))$ is shown in Fig.~\ref{Fig15}.
There exists a state $(9,\emptyset)$ in $Cc(G,Obs(G_{dss}))$.
Then by Lemma~\ref{lemma:4.1}, $G$ is not SCSO w.r.t. $\Sigma_o$ and $X_S$.

\begin{figure}[!ht]
  \centering
    	\begin{tikzpicture}[>=stealth',shorten >=1pt,auto,node distance=2.0 cm, scale = 0.5, transform shape,
	>=stealth,inner sep=2pt]
	\tikzstyle{emptynode}=[inner sep=0,outer sep=0]

	\node[initial, initial where = left, state] (0) {$0$};
	\node[state] (2) [right of =0] {$2$};
	\node[state] (4) [right of =2] {$4$};
	\node[state] (7) [right of =4] {$7$};
	\node[state] (5) [below of =4] {$\red 5$};
	\node[state] (9) [below of =7] {$\red 9$};
	\node[emptynode] (empty1) [below of = 0] {};
	\node[initial, initial where = left, state] (1) [below of =empty1] {$\red 1$};
	\node[state] (3) [right of =1] {$3$};
	\node[state] (6) [right of =3] {$6$};
	\node[state] (8) [right of =6] {$8$};

	\path [->]
	(0) edge node [above, sloped] {$u$} (2)
	(2) edge node [above, sloped] {$b$} (4)
	(4) edge node [above, sloped] {$a$} (7)
	(4) edge node {$a$} (5)
	(7) edge node {$b$} (9)
	(1) edge node {$b$} (3)
	(3) edge node {$a$} (5)
	(3) edge node {$u$} (6)
	(5) edge node {$v$} (6)
	(6) edge node {$b$} (8)
	;

    \end{tikzpicture}
  \caption{The automaton $G$ considered in Example~\ref{ex:4.2}.}
  \label{Fig14}
\end{figure}
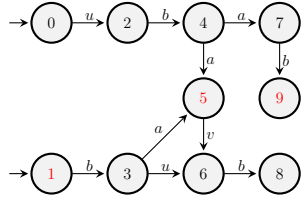
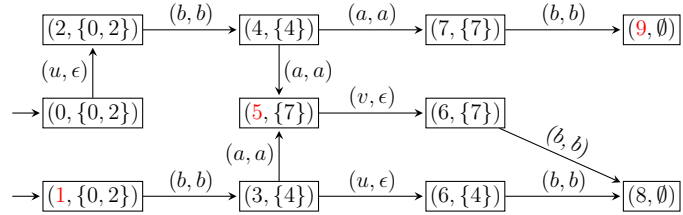
\begin{figure}[!ht]
  \centering
    	\begin{tikzpicture}[>=stealth',shorten >=1pt,auto,node distance=3.5 cm, scale = 0.7, transform shape,
	>=stealth,inner sep=2pt]

	\node[initial, initial where = left, rectangle state] (0-02) {$(0,\{0,2\})$};
	\node[rectangle state] (2-02) [above = 1cm of 0-02] {$(2,\{0,2\})$};
	\node[initial, initial where = left, rectangle state] (1-02) [below = 1cm of 0-02] {$({\red 1},\{0,2\})$};
	\node[rectangle state] (5-7) [right of = 0-02] {$({\red 5},\{7\})$};
	\node[rectangle state] (4-4) [above = 1cm of 5-7] {$(4,\{4\})$};
	\node[rectangle state] (3-4) [below = 1cm of 5-7] {$(3,\{4\})$};
	\node[rectangle state] (7-7) [right of = 4-4] {$(7,\{7\})$};
	\node[rectangle state] (6-7) [right of = 5-7] {$(6,\{7\})$};
	\node[rectangle state] (6-4) [below = 1cm of 6-7] {$(6,\{4\})$};
	\node[rectangle state] (9-) [right of = 7-7] {$({\red 9},\emptyset)$};
	\node[rectangle state] (8-) [right of = 6-4] {$(8,\emptyset)$};

	\path [->]
	(4-4) edge node {$(a,a)$} (5-7)
	(2-02) edge node {$(b,b)$} (4-4)
	(3-4) edge node {$(a,a)$} (5-7)
	(1-02) edge node {$(b,b)$} (3-4)
	(4-4) edge node {$(a,a)$} (7-7)
	(5-7) edge node {$(v,\epsilon)$} (6-7)
	(3-4) edge node {$(u,\epsilon)$} (6-4)
	(0-02) edge node {$(u,\epsilon)$} (2-02)
	(7-7) edge node {$(b,b)$} (9-)
	(6-7) edge node [sloped, above] {$(b,b)$} (8-)
	(6-4) edge node {$(b,b)$} (8-)
	;

  \end{tikzpicture}
  \caption{The concurrent composition $Cc(G,Obs(G_{dss}))$ for the automaton $G$ shown in Fig.~\ref{Fig14}.}
  \label{Fig15}
\end{figure}

We now enforce SCSO for the system $G$ shown in Fig.~\ref{Fig14} using Algorithm~\ref{algo:3}.
In order to prevent state $(9,\emptyset)$ in $Cc(G,Obs(G_{dss}))$ from being reached, we choose controllable transition $(4,\{4\})\stackrel{(a,a)}{\longrightarrow}(7,\{7\})$ to disable since transition $(7,\{7\})\stackrel{(b,b)}{\longrightarrow}(9,\emptyset)$ is uncontrollable.
After disabling transition $4\stackrel{a}{\rightarrow}7$ in $G$, we obtain a new concurrent composition $Cc(G^1,Obs(G^1_{dss}))$ shown in Fig.~\ref{Fig16}.
Further, we choose controllable transitions $(3,\{4\})\stackrel{(a,a)}{\longrightarrow}(5,\emptyset)$ and $(4,\{4\})\stackrel{(a,a)}{\longrightarrow}(5,\emptyset)$ in $Cc(G^1,Obs(G^1_{dss}))$ to disable to prevent state $(5,\emptyset)$ from being reached.
After disabling them, $Cc(G^1,Obs(G^1_{dss}))$ becomes $Cc(G^2,Obs(G^2_{dss}))$ shown in Fig.~\ref{Fig17}.
We see that there exists no state of the form $(x,\emptyset)$ in $Cc(G^2,Obs(G^2_{dss}))$, where $x\in X_S$.
Thus, Algorithm~\ref{algo:3} outputs ``Yes" and the subset $\mathcal{E}_c=\{4\stackrel{a}{\rightarrow}7, 3\stackrel{a}{\rightarrow}5, 4\stackrel{a}{\rightarrow}5\}$.
Hence, subsystem $G_{\backslash\mathcal{E}_c}$ is SCSO w.r.t. $\Sigma_o$ and $X_S$.

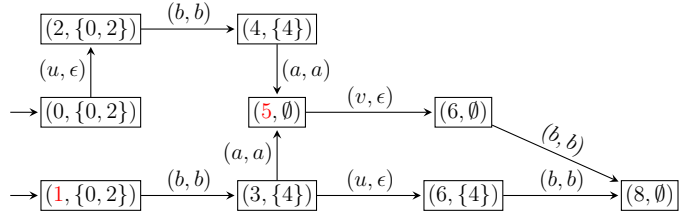
\begin{figure}[!ht]
  \centering
      	\begin{tikzpicture}[>=stealth',shorten >=1pt,auto,node distance=3.5 cm, scale = 0.7, transform shape,
	>=stealth,inner sep=2pt]

	\node[initial, initial where = left, rectangle state] (0-02) {$(0,\{0,2\})$};
	\node[rectangle state] (2-02) [above = 1cm of 0-02] {$(2,\{0,2\})$};
	\node[initial, initial where = left, rectangle state] (1-02) [below = 1cm of 0-02] {$({\red 1},\{0,2\})$};
	\node[rectangle state] (5-7) [right of = 0-02] {$({\red 5},\emptyset)$};
	\node[rectangle state] (4-4) [above = 1cm of 5-7] {$(4,\{4\})$};
	\node[rectangle state] (3-4) [below = 1cm of 5-7] {$(3,\{4\})$};
	\node[rectangle state] (6-7) [right of = 5-7] {$(6,\emptyset)$};
	\node[rectangle state] (6-4) [below = 1cm of 6-7] {$(6,\{4\})$};
	\node[rectangle state] (8-) [right of = 6-4] {$(8,\emptyset)$};

	\path [->]
	(4-4) edge node {$(a,a)$} (5-7)
	(2-02) edge node {$(b,b)$} (4-4)
	(3-4) edge node {$(a,a)$} (5-7)
	(1-02) edge node {$(b,b)$} (3-4)
	(5-7) edge node {$(v,\epsilon)$} (6-7)
	(3-4) edge node {$(u,\epsilon)$} (6-4)
	(0-02) edge node {$(u,\epsilon)$} (2-02)
	(6-7) edge node [sloped, above] {$(b,b)$} (8-)
	(6-4) edge node {$(b,b)$} (8-)
	;

  \end{tikzpicture}
  \caption{The concurrent composition $Cc(G^1,Obs(G^1_{dss}))$.}
  \label{Fig16}
\end{figure}

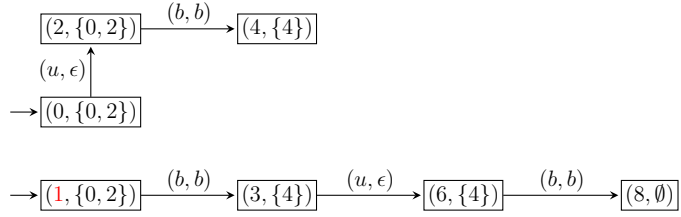
\begin{figure}[!ht]
  \centering
        	\begin{tikzpicture}[>=stealth',shorten >=1pt,auto,node distance=3.5 cm, scale = 0.7, transform shape,
	>=stealth,inner sep=2pt]

	\node[initial, initial where = left, rectangle state] (0-02) {$(0,\{0,2\})$};
	\node[rectangle state] (2-02) [above = 1cm of 0-02] {$(2,\{0,2\})$};
	\node[initial, initial where = left, rectangle state] (1-02) [below = 1cm of 0-02] {$({\red 1},\{0,2\})$};
	\node[rectangle state] (4-4) [above = 1cm of 5-7] {$(4,\{4\})$};
	\node[rectangle state] (3-4) [below = 1cm of 5-7] {$(3,\{4\})$};
	\node[rectangle state] (6-4) [below = 1cm of 6-7] {$(6,\{4\})$};
	\node[rectangle state] (8-) [right of = 6-4] {$(8,\emptyset)$};

	\path [->]
	(2-02) edge node {$(b,b)$} (4-4)
	(1-02) edge node {$(b,b)$} (3-4)
	(3-4) edge node {$(u,\epsilon)$} (6-4)
	(0-02) edge node {$(u,\epsilon)$} (2-02)
	(6-4) edge node {$(b,b)$} (8-)
	;

  \end{tikzpicture}
  \caption{The concurrent composition $Cc(G^2,Obs(G^2_{dss}))$.}
  \label{Fig17}
\end{figure}
\end{example}

\begin{remark}\label{re:4.4}
We discuss the time complexity of using Algorithm~\ref{algo:3} to enforce SCSO.
By~\cite{Han(2023)}, the time cost for computing $Cc(G,Obs(G_{dss}))$ is $\mathcal{O}((|\Sigma_o||\Sigma_{uo}|+|\Sigma|)|X|^22^{|X|})$.
In Line $3$, the time consumption for finding all states of form $(x,\emptyset)$ in $Cc(G,Obs(G_{dss}))$ is $\mathcal{O}(|X|2^{|X|})$.
In Line $11$ (resp., $15$), computing an uncontrollable leaking-secret run in $Cc(G,Obs(G_{dss}))$ from $X_{cc,0}$ to $\Delta$ (resp., a set $\Omega$) can be done with the time complexity $\mathcal{O}(|\Sigma||X|^22^{|X|}+|X|2^{|X|})$ using the ``Breadth-First Search Algorithm" in~\cite{Cormen(2009)}.
In Lines $16$ to $17$, disabling all controllable transitions in $\Omega$ has complexity $\mathcal{O}(|X|^2|\Sigma_c|)$.
Updating each state in $Ac(Cc(G,Obs(G_{dss})))$ can be done in the size of $Cc(G,Obs(G_{dss}))$.
Taken together, the overall time complexity for enforcing SCSO using Algorithm~\ref{algo:3} is $\mathcal{O}((|\Sigma_o||\Sigma_{uo}|+|\Sigma|)|\Sigma_c||X|^42^{|X|})$.
\end{remark}

By Lemma~\ref{lemma:4.1}, the algorithms of enforcing SISO and Inf-SSO can be designed by making a few slight modifications to Algorithm~\ref{algo:3}, respectively.
Specifically, for the enforcement of SISO:

\begin{itemize}
  \item In Algorithm~\ref{algo:3}, ``SCSO" is replaced by ``SISO".
  \item Line $3$ of Algorithm~\ref{algo:3} is replaced by: Collect all states of the form $(\cdot,\emptyset)$ in $Cc(G,Obs(G_{dss}))$ that are reachable from $X_{cc,0}^{S}$ as a set $\Delta$.
  \item In Line $10$ of Algorithm~\ref{algo:3}, we replace state $(x_,\emptyset)$ with state $(\cdot,\emptyset)$.
  \item Line $11$ of Algorithm~\ref{algo:3} is replaced by: there exists an uncontrollable leaking-secret run $(x_0, q)\stackrel{e}{\rightarrow}(\cdot,\emptyset)$ in $Cc(G,Obs(G_{dss}))$, where $(x_0,q)\in X_{cc,0}^{S}$.
  \item In Line $15$ of Algorithm~\ref{algo:3}, we replace $X_{cc,0}$ with $X_{cc,0}^{S}$.
\end{itemize}

Whereas for the enforcement of Inf-SSO: in Algorithm~\ref{algo:3} we only need to replace ``SCSO" and state $(x,\emptyset)$ (where $x\in X_S$) with ``Inf-SSO" and state $(\cdot,\emptyset)$, respectively.

\begin{theorem}\label{th:4.5}
A subset $\mathcal{E}_c\subseteq\mathcal{T}_c^G$ returned by the modified Algorithm~\ref{algo:3} (if exists) is a solution of Problem~\ref{problem:1} on SISO (resp., Inf-SSO).
\end{theorem}

\begin{proof}
This proof is similar to that of Theorem~\ref{th:4.3}.
We here omit it.
\end{proof}

\begin{theorem}\label{th:4.6}
Given a system $G=(X,\Sigma,\delta,X_0)$, a set $\Sigma_o\subseteq\Sigma$ of observable events, a set $\Sigma_c\subseteq\Sigma$ of controllable events, and a set $X_S\subseteq X$ of secret states,
if the modified Algorithm~\ref{algo:3} outputs ``No" (i.e., no subset $\mathcal{E}_c\subseteq\mathcal{T}_c^G$ is returned by the modified Algorithm~\ref{algo:3}), then Problem~\ref{problem:1} on SISO (resp., Inf-SSO) has no solution.
\end{theorem}

\begin{proof}
This proof is similar to that of Theorem~\ref{th:4.2}.
We also omit it here.
\end{proof}

\begin{remark}\label{re:4.5}
The (worst-case) time complexity of enforcing SISO (resp., Inf-SSO) using the modified Algorithm~\ref{algo:3} also is $\mathcal{O}((|\Sigma_o||\Sigma_{uo}|+|\Sigma|)|\Sigma_c||X|^42^{|X|})$.
Note that, the authors in~\cite{Ma(2021)} considered the enforcement of Inf-SSO under supervisory control framework using so-called infinite-step recognizer ($\infty$-SR).
In particular, when $\Sigma_c\subseteq\Sigma_o$ holds, a supervisor for enforcing Inf-SSO can be obtained by modifying $\infty$-SR whose time complexity is $\mathcal{O}(|\Sigma_{o}||\Sigma_{uo}||X|2^{2|X|})$.
Therefore, our proposed approach for enforcing Inf-SSO is more efficient than that of~\cite{Ma(2021)} from the complexity's point of view.
\end{remark}

\begin{example}\label{ex:4.3}
Let us consider again the system $G$ shown in Fig.~\ref{Fig14}, whose concurrent composition $Cc(G,Obs(G_{dss}))$ is shown in Fig.~\ref{Fig15}.
By Lemma~\ref{lemma:4.1}, $G$ is neither SISO nor Inf-SSO.

We firstly enforce SISO of $G$.
Specifically, state $(8,\emptyset)$ in $Cc(G,Obs(G_{dss}))$ violates the notion of SISO.
Hence, we cut off all runs from initial state $(1,\{0,2\})$ to state $(8,\emptyset)$ by disabling the last controllable transition of each of all such runs to prevent state $(8,\emptyset)$ from being reached.
After disabling controllable transitions $(3,\{4\})\stackrel{(u,\epsilon)}{\longrightarrow}(6,\{4\})$ and $(5,\{7\})\stackrel{(v,\epsilon)}{\longrightarrow}(6,\{7\})$, $Cc(G,Obs(G_{dss}))$ becomes $Cc(G^3,$ $Obs(G^3_{dss}))$, as shown in Fig.~\ref{Fig18}, from which we see that there exists no state of violating SISO.
Therefore, subsystem $G^3=Ac(G_{\backslash\mathcal{E}_c})$ is SISO w.r.t. $\Sigma_o$ and $X_S$, where $\mathcal{E}_c=\{3\stackrel{u}{\rightarrow}6, 5\stackrel{v}{\rightarrow}6\}$.

\begin{figure}[!ht]
  \centering
      	\begin{tikzpicture}[>=stealth',shorten >=1pt,auto,node distance=3.5 cm, scale = 0.7, transform shape,
	>=stealth,inner sep=2pt]

	\node[initial, initial where = left, rectangle state] (0-02) {$(0,\{0,2\})$};
	\node[rectangle state] (2-02) [above = 1cm of 0-02] {$(2,\{0,2\})$};
	\node[initial, initial where = left, rectangle state] (1-02) [below = 1cm of 0-02] {$({\red 1},\{0,2\})$};
	\node[rectangle state] (5-7) [right of = 0-02] {$({\red 5},\{7\})$};
	\node[rectangle state] (4-4) [above = 1cm of 5-7] {$(4,\{4\})$};
	\node[rectangle state] (3-4) [below = 1cm of 5-7] {$(3,\{4\})$};
	\node[rectangle state] (7-7) [right of = 4-4] {$(7,\{7\})$};
	\node[rectangle state] (9-) [right of = 7-7] {$({\red 9},\emptyset)$};

	\path [->]
	(4-4) edge node {$(a,a)$} (5-7)
	(2-02) edge node {$(b,b)$} (4-4)
	(3-4) edge node {$(a,a)$} (5-7)
	(1-02) edge node {$(b,b)$} (3-4)
	(4-4) edge node {$(a,a)$} (7-7)
	(0-02) edge node {$(u,\epsilon)$} (2-02)
	(7-7) edge node {$(b,b)$} (9-)
	;

  \end{tikzpicture}
  \caption{A concurrent composition $Cc(G^3,Obs(G^3_{dss}))$.}
  \label{Fig18}
\end{figure}
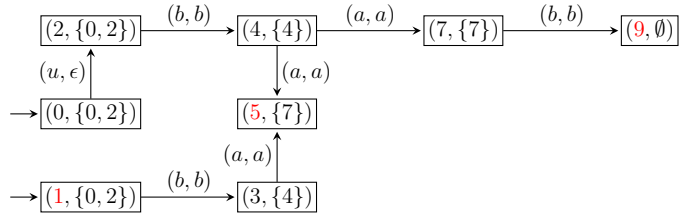

Secondly, we enforce Inf-SSO of $G$.
In order to prevent states $(8,\emptyset)$ and $(9,\emptyset)$ in $Cc(G,Obs(G_{dss}))$ that violate Inf-SSO from being reached, we disable all last controllable transitions of all such runs ended with state $(8,\emptyset)$ or $(9,\emptyset)$, which are $(3,\{4\})\stackrel{(u,\epsilon)}{\longrightarrow}(6,\{4\})$, $(5,\{7\})\stackrel{(v,\epsilon)}{\longrightarrow}(6,\{7\})$, and $(4,\{4\})\stackrel{(a,a)}{\longrightarrow}(7,\{7\})$.
This results in a new concurrent composition $Cc(G^4,Obs(G^4_{dss}))$ shown in Fig.~\ref{Fig19}.
State $(5,\emptyset)$ in $Cc(G^4,Obs(G^4_{dss}))$ also violates Inf-SSO.
Hence, we disable controllable transitions $(3,\{4\})\stackrel{(a,a)}{\longrightarrow}(5,\emptyset)$ and $(4,\{4\})\stackrel{(a,a)}{\longrightarrow}(5,\emptyset)$ to cut off all such runs ended with state $(5,\emptyset)$.
After disabling them, we obtain the modified concurrent composition $Cc(G^5,Obs(G^5_{dss}))$ shown in Fig.~\ref{Fig20}.
By Lemma~\ref{lemma:4.1}, we have that subsystem $G^5=Ac(G_{\backslash\mathcal{E}_c})$ is Inf-SSO w.r.t. $\Sigma_o$ and $X_S$, where $\mathcal{E}_c=\{3\stackrel{u}{\rightarrow}6, 5\stackrel{v}{\rightarrow}6, 4\stackrel{a}{\rightarrow}7, 3\stackrel{a}{\rightarrow}5, 4\stackrel{a}{\rightarrow}5\}$.

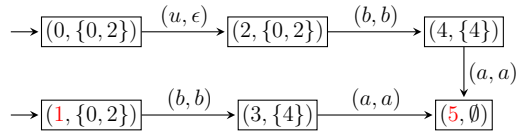
\begin{figure}[!ht]
  \centering
    \begin{tikzpicture}[>=stealth',shorten >=1pt,auto,node distance=3.5 cm, scale = 0.7, transform shape,
	>=stealth,inner sep=2pt]

	\node[initial, initial where = left, rectangle state] (0-02) {$(0,\{0,2\})$};
	\node[rectangle state] (2-02) [right of = 0-02] {$(2,\{0,2\})$};
	\node[initial, initial where = left, rectangle state] (1-02) [below = 1cm of 0-02] {$({\red 1},\{0,2\})$};

	\node[rectangle state] (4-4) [right of = 2-02] {$(4,\{4\})$};
	\node[rectangle state] (3-4) [right of = 1-02] {$(3,\{4\})$};
	\node[rectangle state] (5-7) [right of = 3-4] {$({\red 5},\emptyset)$};

	\path [->]
	(4-4) edge node {$(a,a)$} (5-7)
	(2-02) edge node {$(b,b)$} (4-4)
	(3-4) edge node {$(a,a)$} (5-7)
	(1-02) edge node {$(b,b)$} (3-4)
	(0-02) edge node {$(u,\epsilon)$} (2-02)
	;

  \end{tikzpicture}
  \caption{A concurrent composition $Cc(G^4,Obs(G^4_{dss}))$.}
  \label{Fig19}
\end{figure}
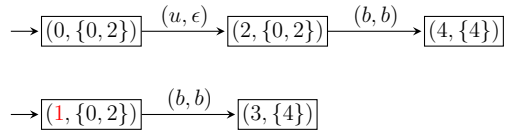
\begin{figure}[!ht]
  \centering
  \begin{tikzpicture}[>=stealth',shorten >=1pt,auto,node distance=3.5 cm, scale = 0.7, transform shape,
	>=stealth,inner sep=2pt]

	\node[initial, initial where = left, rectangle state] (0-02) {$(0,\{0,2\})$};
	\node[rectangle state] (2-02) [right of = 0-02] {$(2,\{0,2\})$};
	\node[initial, initial where = left, rectangle state] (1-02) [below = 1cm of 0-02] {$({\red 1},\{0,2\})$};

	\node[rectangle state] (4-4) [right of = 2-02] {$(4,\{4\})$};
	\node[rectangle state] (3-4) [right of = 1-02] {$(3,\{4\})$};

	\path [->]
	(2-02) edge node {$(b,b)$} (4-4)
	(1-02) edge node {$(b,b)$} (3-4)
	(0-02) edge node {$(u,\epsilon)$} (2-02)
	;

  \end{tikzpicture}
  \caption{A concurrent composition $Cc(G^5,Obs(G^5_{dss}))$.}
  \label{Fig20}
\end{figure}
\end{example}

\section{Concluding remarks}\label{sec5}

In this paper, we studied the verification and enforcement of strong state-based opacity for partially-observed discrete-event systems.
In order to efficiently verify if a given system $G$ is $K$-SSO in Definition~\ref{de:3.1}, we proposed a new concurrent-composition structure $Cc(\hat{G},\tilde{G}_{obs})$, which is a slight variant of $Cc(G,Obs(G_{dss}))$ proposed in our previous work~\cite{Han(2023)}.
Using this variant structure, we designed an algorithm for verifying $K$-SSO in Definition~\ref{de:3.1}, which does not depend on $K$ if $K>|\hat{X}|2^{|X\backslash X_S|}$, resulting in an upper bound of $|\hat{X}|2^{|X\backslash X_S|}-1$ on $K$ in $K$-SSO in Definition~\ref{de:3.1}.

We proposed a distinctive opacity-enforcement mechanism that chooses controllable transitions to disable (if possible) before an original system of interest starts to run in order to make system's secrets from non-opaque to opaque, which remarkably differs from the existing opacity-enforcement approaches in the literature.
To enforce multifarious types of strong state-based opacity, we designed the corresponding algorithms using the proposed $Cc(\hat{G},\tilde{G}_{obs})$ and $Cc(G,Obs(G_{dss}))$, respectively.
In particular, the proposed algorithm for enforcing Inf-SSO is better than that of~\cite{Ma(2021)} from the complexity's point of view.

The proposed opacity-enforcement mechanism, in this paper, is to restrict the original system's behavior to ensure that the system does not reveal its ``secrets" to an intruder.
However, it does not apply to such a scenario where the system must execute its full behavior.
To overcome this drawback, it is reasonable to extend the proposed opacity-enforcement approach by disabling, adding, and/or replacing transitions without creating new behavior.
Thus, an intruder still cannot learn for sure whether the structure of the original system has been modified based on his/her full knowledge of the system structure.
The interesting direction is left for future study.


\bibliographystyle{plain}

\end{document}